 \documentclass[journal]{IEEEtran}
\usepackage{graphicx}
\usepackage{setspace}
\usepackage{amsmath}
\usepackage{amssymb} 
\usepackage{bm} 
\usepackage[normalem]{ulem}
\usepackage{cite}
\usepackage{float}
\usepackage[usenames,dvipsnames]{pstricks}
\usepackage{epsfig}
\usepackage{pst-grad} % For gradients
\usepackage{pst-plot} % For axes
\usepackage{tabularx}
\usepackage[T1]{fontenc}
\newcommand{\lb} {\left}
\newcommand{\rb} {\right}
\newcommand{\nn} {\nonumber}
\newcommand{\Ei}{\mathrm{Ei}}
\newtheorem{lemma}{Lemma}

\newenvironment{proof}[1][Proof:]{\begin{trivlist}
\item[\hskip \labelsep {\itshape #1}]}{\end{trivlist}}

\allowdisplaybreaks
\usepackage{lipsum}
\usepackage{flushend}
\flushend

\usepackage{caption}
\usepackage{subcaption}
\usepackage[utf8]{inputenc}

\IEEEaftertitletext{\vspace{-2.4\baselineskip}}
\begin{document}
\onecolumn{\noindent © 2023. Personal use of this material is permitted. Permission from authors must be obtained for all other uses, in any current or future media, including reprinting/republishing this material for advertising or promotional purposes, creating new collective works, for resale or redistribution to servers or lists, or reuse of any copyrighted component of this work in other works.}
 \twocolumn{
\title{Transmitter Selection for Secrecy Against  Colluding Eavesdroppers with Backhaul Uncertainty}

\author{Burhan Wafai,~\IEEEmembership{Student Member,~IEEE}, Ankit Dubey,~\IEEEmembership{Member,~IEEE}, and  Chinmoy Kundu,~\IEEEmembership{Member,~IEEE}
% \corresp{Corresponding author:  Chinmoy Kundu (e-mail: chinmoy.kundu@ucd.ie,).}
\thanks{Burhan Wafai and Ankit Dubey are with Indian Institute of Technology Jammu, India
(email: burhan.wafai@iitjammu.ac.in and ankit.dubey@iitjammu.ac.in).}
\thanks{Chinmoy Kundu is with University College Dublin, Ireland (email: chinmoy.kundu@ucd.ie).}

\thanks{This publication has emanated from research supported in part by the Science Foundation Ireland (SFI) under Grant Number 17/US/3445 and 22/IRDIFA/10425, the Tata Consultancy Services (TCS) Foundation through its TCS Research Scholar Program, and by the Department of Science and Technology (DST), India sponsored project TMD/CERI/BEE/2016/059.
} 

}
\maketitle
\thispagestyle{empty}
\pagestyle{empty}
\pagestyle{plain} 
\begin{abstract}
Due to the exponential growth of interconnected devices and reduced cell coverage, beyond fifth-generation networks will be dense. Thus, instead of wired backhaul, wireless backhaul will be cost-effective and flexible. For security in multi-transmitter systems, sub-optimal and optimal transmitter selection schemes exist. However, including backhaul activity knowledge available (BKA) and backhaul activity knowledge unavailable (BKU) cases and transmitter selection schemes, there is no generalized secrecy analysis method. Moreover, evaluation of the ergodic secrecy rate (ESR) of the optimal selection schemes is impossible using existing solution approaches. 
To address these, we propose two sub-optimal and optimal transmitter selection schemes for a small-cell multi-transmitter system in BKU or BKA cases in the presence of multiple colluding eavesdroppers.  
We derive the distribution of the ratio of the destination channel SNR and eavesdropping channel SNR, thereby providing different secrecy performance metrics uniformly irrespective of selection schemes and BKU or BKA cases. 
Simplified asymptotic expressions are provided to elucidate the influence of the system parameters and of the backhaul reliability. 
We observe that the secrecy performance improves when the knowledge of backhaul link activity is utilized, and the improvement is most noticeable when the backhaul is highly unreliable. We also observe that while the secrecy performance degrades with an increasing number of eavesdroppers, neither the asymptotic saturation value of the secrecy outage probability nor the rate of improvement of the ESR with signal-to-noise-ratio depends on the number of eavesdroppers.   
\end{abstract}



\maketitle

\section{INTRODUCTION}
\IEEEPARstart{W}{ireless} network infrastructures in the  fifth-generation (5G) and beyond technologies
%  will be more cond sed  an ever \cite{Osseiran20145G}. They 
will experience an exponential growth of interconnected devices.
% in networks in wireless communications 
% will increase the network's density 
% with many base stations (BSs). To cope with the increase of traffic volume in wireless communications, future wireless technologies will include  many base stations 
(BSs) or access points (APs) making the network highly dense and of heterogeneous \cite{Xiaohu20145G}. 
% The backhaul is a crucial infrastructure that offers the necessary services for the fronthaul network. 
% In such heterogeneous networks, the wireless backhaul links act as the backbone for the densely interconnected nodes  due to the structural support for the densely interconnected nodesThe wireless backhaul links in such heterogeneous networks form the backbone 
for the dense inter-connected nodes since it enables the communication between the AP and several transmitters  \cite{Osseiran20145G}. 
In practical systems, the incorporation of wireless backhaul is an easy, flexible, and cost-efficient alternative to wired backhaul.  However, wireless backhaul connections are unreliable as the backhaul links have a certain probability of failure  \cite{OnurBackhaul2022}. 
% The nodes in such ultra-dense wireless networks  will require  greater connectivity  and thus  need the wireless backhaul links provided by the APs \cite{Xiaohu20145G, OnurBackhaul2022}. 

% The dense nature of the wireless networks also makes the backhaul connections dense in a multi-transmitter system. 
% In future large-scale networks, the traditional wired backhaul links are less favorable compared to wireless backhaul links because of deployment complexity, high cost, and rigidity in the practical systems \cite{Orawan2011Evolution}.

Physical layer security (PLS) has recently emerged as a means of securing wireless networks, which are constantly at risk from eavesdropping attacks \cite{wyner_wiretap},\cite{Gamal_On_the_Sec_Cap_Fad_Ch}. The transmitter selection in a multi-source network can improve the PLS by increasing the diversity gain and hence the secrecy of the system without using multiple antennas \cite{bletsas2006, Kundu17_sopeaves,chinmoy_GC16, Kundu15_dual}.
Based on the availability of the channel state information (CSI), the transmitter selection can be classified into two categories; optimal transmitter selection (OTS) and sub-optimal transmitter selection (STS). An OTS scheme requires global CSI, whereas an STS scheme does not require global CSI \cite{chinmoy_GC16, Vu2017Secure, Kundu15_dual}. Hence, STS schemes can reduce the complexity and power consumption of the network and thereby extending the lifetime of the network. The STS scheme performed by maximizing the destination channel rate without considering the eavesdropping rate has been referred to as the traditional transmitter selection (TTS) scheme in the literature \cite{kim2016secrecy,  Vu2017Secure}. The authors in \cite{Kundu17_sopeaves} proposed an STS scheme by selecting a transmitter corresponding to the worst eavesdropping link without considering the destination rate to improve the secrecy of the system. We refer to this scheme as a minimal eavesdropping selection  (MIN-ES) scheme.

% and thus increases the complexity and power consumption.

The reliability of the backhaul links affects the system performance, therefore, it is important to explore the effect of the backhaul uncertainty on the secrecy performance of transmitter selection schemes which have extensively been studied in  \cite{Yincheng,TAKhan2015Coperative,kimBackhaulSelection, Debbah2012_het_backhaul, Kundu_TVT19, wafaiVTC, kotwal2021transmitter,chinmoy_letter2021,kotwalTVT_backhaul}. Although the knowledge of backhaul link activity is crucial, it might not always be available for transmitter selection \cite{Kundu_TVT19}.  The unavailability of backhaul link activity knowledge, which we denote as the backhaul link activity knowledge unavailable (BKU) case, can result in selecting a transmitter with an inactive backhaul link.  When the backhaul link activity knowledge is available, which we denote as the backhaul link activity knowledge available (BKA) case, the selection can be made from the set of transmitters with active backhaul links \cite{wafaiVTC,kotwal2021transmitter,chinmoy_letter2021,kotwalTVT_backhaul}. The knowledge of active backhaul links improves the secrecy performance.

% In the BKU case, the transmitter selection is performed irrespective of whether the link is active. In the BKA case, the transmitter selection is carried out among the transmitters with active backhaul links. 

In \cite{Kundu_TVT19}, the authors studied the secrecy of the cognitive radio (CR) network in the presence of unreliable backhaul links in the BKU case; however, the 
BKA case was not considered for the transmitter selection.
% knowledge of which backhaul links were active was not available before the transmitter selection. 
The non-zero secrecy rate (NZSR), secrecy outage probability (SOP), and ergodic secrecy rate (ESR) were evaluated for a number of transmitter selection schemes, such as MIN-ES, TTS, OTS, and the transmitter selection scheme with minimal interference.
% With the available knowledge of active backhaul links,  the authors in \cite{wafaiVTC} evaluated the SOP in a CR network for the TTS and OTS schemes.
The authors in \cite{wafaiVTC} then extended the SOP analysis in the same system for the TTS and OTS schemes for the BKA case.
In \cite{kotwal2021transmitter},  the SOP for the TTS and OTS schemes in a frequency selective fading channel was evaluated for the BKA case;
however, the ESR was not evaluated for the selection schemes. Though, the authors in \cite{kotwalTVT} evaluated the ESR of the OTS
% source-destination pair selection
scheme in the same channel model but did not show the effect of backhaul links.   In all these papers \cite{Kundu_TVT19,wafaiVTC,kotwal2021transmitter,kotwalTVT}, the secrecy performance was evaluated for a single eavesdropper.

In the case of multiple eavesdroppers, the ESR of the OTS scheme was evaluated in \cite{chinmoy_letter2021} for the BKA case under the frequency-flat fading channel condition. 
However, the SOP was not evaluated, and the MIN-ES and TTS schemes were not considered. In \cite{kotwalTVT_backhaul}, the authors evaluated the SOP and ESR for TTS and OTS schemes including multiple eavesdroppers in the BKA case under the frequency-selective fading channels condition. However, the MIN-ES scheme was not studied. Although the authors in \cite{chinmoy_letter2021} and \cite{kotwalTVT_backhaul} considered multiple eavesdroppers, a straightforward non-colluding case was considered where an eavesdropper with maximum instantaneous SNR determines the secrecy performance.  Moreover,  in all the above-mentioned papers \cite{wafaiVTC,kotwal2021transmitter,kotwalTVT,chinmoy_letter2021,kotwalTVT_backhaul}, the MIN-ES scheme was not considered with wireless backhaul links.

To the best of the authors' knowledge, transmitter selection schemes with colluding eavesdroppers performing maximal ratio combining (MRC) have not been explored for both the backhaul link activity knowledge cases (BKU and BKA). Although the secrecy performance of networks with wireless backhaul links has been widely studied in the above-mentioned works, an effort to provide a generalized mathematical framework that incorporates both the backhaul uncertainty (BKU and BKA) cases irrespective of different transmitter selection schemes leading to all the secrecy performance metrics (NZSR, SOP, and ESR) is missing from the literature.

Motivated by the above discussion, we consider a wireless backhaul-aided network with multiple transmitters and a single destination in the presence of multiple colluding eavesdroppers. 
Two cases of backhaul link activity knowledge, BKU and BKA cases are considered.  To enhance the system's secrecy performance, we investigate MIN-ES, TTS, and OTS schemes under BKU and BKA cases.
The MIN-ES method has not been studied by \cite{chinmoy_letter2021,wafaiVTC,kotwal2021transmitter,kotwalTVT_backhaul} for secrecy including wireless backhaul. We study the case of colluding eavesdroppers where the eavesdroppers utilize the MRC technique as opposed to  \cite{chinmoy_letter2021,kotwalTVT_backhaul}. Moreover, in \cite{chinmoy_letter2021}, the authors did not consider the SOP analysis and the BKU case both of which we incorporate. The authors in \cite{kotwalTVT_backhaul} considered frequency-selective channels whereas we assume frequency-flat fading channels as a result the analysis is entirely different. We also present a generalized approach to incorporate both backhaul uncertainty (BKU and BKA) cases into the performance analysis by incorporating a mixture distribution method. Our secrecy performance analysis approach by finding the distribution of the ratio of the destination channel and eavesdropping channel SNR uniformly provides the NZSR, SOP, and ESR performances.

The main contributions of this paper are listed as follows:
\begin{itemize}

    \item We evaluate the exact closed-form expressions of the NZSR, SOP, and ESR for the two sub-optimal transmitter selections (MIN-ES and TTS) and the optimal transmitter selection (OTS) schemes in both the BKU and BKA cases. 
 
    \item We consider multiple colluding eavesdroppers performing MRC under independent but non-identical (INID) channels.  
    
     \item To obtain greater insights, simplified asymptotic expressions for the NZSR, SOP, and ESR under both BKU and BKA cases and diversity order for the SOP under perfect backhaul conditions are also included.
     
     \item A general and unified method is provided to incorporate backhaul uncertainty using a mixture distribution and the ratio of destination and eavesdropper channel SNRs which is then directly used to find the SOP, NZSR, and ESR, along with their asymptotes.
     
\end{itemize}
% 

% advantage which arises from using these techniques.

\color{black}
The rest of the paper is organized as follows:  system model is described in Section \ref{system model}  followed by distribution of $\Gamma_{\text{R}}^{(n^*)}$ for each selection scheme in Section \ref{sec_distribution_snr_ratio}. The ESR is discussed in Section \ref{section_ergodic_secrecy_rate}  and the asymptotic SOP analysis for unreliable backhaul links is provided in Section \ref{sec_asymptotic_analysis_sop_unreliable}. Section \ref{section_diversity_order} discusses the diversity order for perfect backhaul links and Section \ref{section_asymptotic_ESR} provides the asymptotic ESR. Finally, the numerical results with discussions are presented in Section \ref{section_numerical_results} followed by conclusions in Section \ref{conclusions}.

\textit{Notation:} The probability of occurrence of an event and the expectation of a random variable (RV) $X$ are denoted by $\mathbb{P}[\cdot]$  and $\mathbb{E}_X[\cdot]$, respectively.
The channel coefficient of a link between nodes $\text{A}$ and $\text{B}$ is denoted by $h^{(n)}_{\text{AB}}$,  the end-to-end signal-to-noise-ratio (SNR) which incorporates the backhaul link reliability factor at $\text{B}$  for the link $\text{A-B}$ is denoted as $\hat{\Gamma}_{\text{AB}}$ while as the SNR without the inclusion of backhaul link reliability factor is denoted as ${\Gamma}_{\text{AB}}$. The cumulative distribution function (CDF) of an RV $X$ is denoted by $F_{X}(\cdot)$. Correspondingly, $f_{X}(\cdot)$ is the respective probability density functions (PDFs). 
\section{System and channel model}\label{system model}
We consider a network where multiple transmitters are connected to an access point AP via wireless backhaul links, as presented in Fig. \ref{fig:SM2}. The network consists of $N$ transmitters $\text{S}_n$, where $n\in\{1, \ldots, N\}$, serving an user $\text{D}$ in presence of $K$ eavesdroppers $\text{E}_k$ where $k\in\{1, \ldots, K\}$. Each node in the system is assumed to be equipped with a single antenna. Wireless backhaul links are present between each transmitter and the AP which have a certain probability of link failure.
The backhaul link reliability between AP and $\text{S}_n$ for each $n\in\{1,\ldots N\}$ is modeled by independent Bernoulli RV $\mathbb{I}_n \in\{0,1\}$ with backhaul link reliability factor  $\mathbb{P}[\mathbb{I}_n=1]=s$ and backhaul link failure probability $\mathbb{P}[\mathbb{I}_n=0]=(1-s)$ where $0<s\le 1$ is the backhaul link reliability factor.

\begin{figure}
\centering 
\includegraphics[width=2.7in]{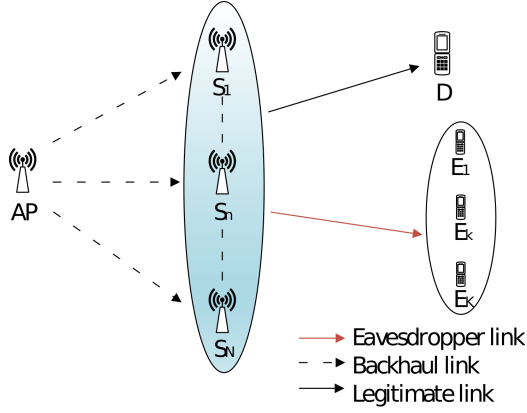} 
\caption{
A multi-transmitter network with unreliable backhaul connections.
}
\vspace{-0.5cm}
\label{fig:SM2}
\end{figure}

The channel gains $h^{(n)}_{\text{SD}}$ of the $\text{S}_n$-$\text{D}$ link for each $n$ is assumed to be independent and identically distributed (i.i.d) exhibiting frequency flat Rayleigh fading.   Since each link is Rayleigh distributed, its power gain $|h^{(n)}_{\text{SD}}|^2$ is exponentially distributed. We assume that $|h^{(n)}_{\text{SD}}|^2$ for each $n$ has a power gain parameter $\lambda_{\text{D}}$. 
% Without loss of generality, we also assume that the noise power affecting each receiver is normalized to unity, thus the average received SNR ffor the link can be written as $1/\lambda_{\text{D}}$.
% of $\text{S}_n$-$\text{D}$ link.
% , respectively  (not necessarily $\lambda. _{D}=\lambda_{E}$).
Without loss of generality, assuming identical additive white Gaussian noise (AWGN) parameters at each receiver with zero mean and unit variance, the instantaneous SNR of the  $\text{S}_n$-$\text{D}$ link is denoted as $\Gamma^{(n)}_{\text{SD}}=|h^{(n)}_{\text{SD}}|^2$. The corresponding PDF and CDF of $\Gamma^{(n)}_{\text{SD}}$ are expressed as
\begin{align}\label{pdf_SD}
f_{\Gamma^{(n)}_{\text{SD}}}(x)&=\lambda_{\text{D}}\exp(-\lambda_\text{D}x),\\ \label{cdf_SD}
F_{\Gamma^{(n)}_{\text{SD}}}(x)&=1-\exp(-\lambda_\text{D}x),    
\end{align}
respectively.

All $\text{S}_n$-$\text{E}_k$ links for each  $k\in\{1,\ldots, K\}$ and a particular $n\in\{1,\ldots, N\}$ are assumed to be independent non-identically distributed (i.n.i.d) Rayleigh fading channels with power gain parameter $\lambda_{\text{E}}^{(k)}$. For a given $k\in\{1,\ldots, K\}$, all $\text{S}_n$-$\text{E}_k$ links for each $n\in\{1,\ldots, N\}$ are assumed to be i.i.d. All the eavesdroppers corresponding to a transmitter $\text{S}_n$ collude with each other by performing  MRC technique to maximize the eavesdropping SNR as opposed to \cite{chinmoy_letter2021, kotwalTVT_backhaul}, where the worst eavesdropping SNR determines the secrecy performance. Therefore, the PDF and CDF of the eavesdropping SNR $\Gamma^{(n)}_{\text{SE}}$ corresponding to $\text{S}_n$ is written as \cite{MRC_Akkouchi}
\begin{align} \label{PDF_SE}
&f_{\Gamma^{(n)}_{\text{SE}}}(x)=\sum_{k=1}^{K}\frac{\prod_{m=1}^{K}\lambda_{\text{E}}^{(m)}}{\prod\limits_{\substack{j=1\\j\ne k}}^K(\lambda_{\text{E}}^{(j)}-\lambda_{\text{E}}^{(k)})} e^{-\lambda_{\text{E}}^{(k)} x},
\\
% \end{align}
% %and the corresponding CDF will be
% \begin{align} 
\label{CDF_SE}
&F_{\Gamma^{(n)}_{\text{SE}}}(x)=1-\sum_{k=1}^{K}\frac{\prod_{m=1}^{K}\lambda_{\text{E}}^{(m)}}{\lambda_{\text{E}}^{(k)}\prod\limits_{\substack{j=1\\j\ne k}}^K
(\lambda_{\text{E}}^{(j)}-\lambda_{\text{E}}^{(k)})} e^{-\lambda_{\text{E}}^{(k)} x},
\end{align}
respectively, when 
 $\lambda_{\text{E}}^{(j)}\ne\lambda_{\text{E}}^{(k)}$ (as we assume  i.n.i.d eavesdropping links, there is a high probability that this is maintained).
We evaluate the non-zero secrecy rate (NZSR), secrecy outage probability (SOP), and ergodic secrecy rate (ESR) for each transmitter selection scheme in this work, including backhaul link reliability factor  $s$ taken into consideration. Towards this, we first define the achievable secrecy rate corresponding to $\text{S}_n$ in bits per channel use when the corresponding backhaul link is active (i.e., $s=1$) as
% \begin{align} \label{secrecy_capacity}
% C^{(n)}_{S}=\max\Big\{\log_2\Big(\Gamma_{\text{R}}^{(n^*)}\Big),0\Big\}.  
% \end{align}
% where ${\Gamma_{\text{R}}^{(n^*)}}=\frac{1+{\Gamma}_{\text{SD}}^{(n)}}{1+{\Gamma}_{\text{SE}}^{(n)}}$.

\begin{align} \label{secrecy_capacity}
C^{(n)}_{S}=\max\lb\{\log_2\Big(\Gamma_{\text{R}}^{(n)}\Big),0\rb\}.  
\end{align}
where ${\Gamma_{\text{R}}^{(n)}}=\frac{1+{\Gamma}_{\text{SD}}^{(n)}}{1+{\Gamma}_{\text{SE}}^{(n)}}$. The distribution of ${\Gamma_{\text{R}}^{(n)}}$ is obtained as
\begin{align}\label{gamma_r}
F_{\Gamma_{\text{R}}^{(n)}}(x)&=\mathbb{P}\Bigg[\frac{1+{\Gamma}_{\text{SD}}^{(n)}}{1+{\Gamma}_{\text{SE}}^{(n)}}\le x\Bigg]\nn\\
&=\int_{0}^{\infty} F_{{\Gamma}^{(n)}_{\text{SD}}}(x(y+1)-1) f_{{\Gamma}_{\text{SE}}^{(n)}}(y) dy.
\end{align}

$F_{\Gamma_{\text{R}}^{(n)}}(x)$ is then directly used to find the NZSR, SOP, and ESR.
The NZSR is defined as the probability that the achievable secrecy rate is greater than zero, and for $\text{S}_n$ when the corresponding backhaul link is active, it is evaluated as \cite{yang2013physical}
\begin{align}
\label{non_zero_0}
\mathcal{P}^{(n)}_{\text{NZ}}&=\mathbb P [C^{(n)}_{S}>0]
% \nn\\
% &
=1-  \int_{0}^{\infty} F_{{\Gamma}^{(n)}_{\text{SD}}}(x) f_{{\Gamma}^{(n)}_{\text{SE}}}(x)dx. 
\end{align}
The SOP for $\text{S}_n$ when the corresponding backhaul link is active is defined as the probability when the achievable secrecy rate $C_{S}^{(n)}$ falls below a threshold secrecy rate $R_{th}$ and is evaluated as \cite{wang2015security,yang2013transmit}
\begin{align}\label{secrecy OP equation}
\mathcal{P}_{\text{out}}^{(n)} (R_{th})
&= \mathbb P\Big[C_{S}^{(n)} \le R_{th}\Big]
% \nn \\
% &=\mathbb{P}\left[\log_2\Big(\frac{1+{\Gamma}_{\text{SD}}^{(n)}}{1+{\Gamma}_{\text{SE}}^{(n)}}\Big)< R_{th}\right], 
\nn\\
% &
&= \int_{0}^{\infty} F_{{\Gamma}^{(n)}_{\text{SD}}}(\rho(x+1)-1) f_{{\Gamma}_{\text{SE}}^{(n)}}(x) dx,
\end{align}
where 
$\rho = 2^{R_{th}}$.
% , and $R_{th}$ is the target rate of the  network. 
The ESR for $\text{S}_n$ when the corresponding backhaul link is active is defined as 
% the average secrecy rate averaged over all the SNR distributions and expressed as 
\cite{chinmoy_letter2021}

\begin{align}\label{ergodic_secrecy_rate}
\mathcal{C}_{\text{erg}}^{(n)}&=\frac{1}{\ln(2)}\int_{1}^{\infty}\log(x)f_{\Gamma_{\text{R}}^{(n)}}(x)dx
\nn \\
&
=\frac{1}{\ln(2)}\int_1^\infty\Bigg(\frac{1-F_{\Gamma_{\text{R}}^{(n)}}(x)}{x}\Bigg)dx.  
\end{align}
It is to be noted from (\ref{gamma_r}) and (\ref{secrecy OP equation}) that the SOP can be obtained from the CDF of $\Gamma_{\text{R}}^{(n)}$ as 
\begin{align}
\label{eq_new_SOP_from_GammaR}
\mathcal{P}_{\text{out}}^{(n)}(R_{th})= F_{\Gamma_{\text{R}}^{(n)}}(\rho),  
\end{align}
and from (\ref{non_zero_0}), (\ref{secrecy OP equation}), and (\ref{eq_new_SOP_from_GammaR}) that the NZSR can also be obtained from the CDF of $\Gamma_{\text{R}}^{(n)}$ or the SOP as \begin{align}
\label{eq_NZSR_from_SOP}
\mathcal{P}^{(n)}_{\text{NZ}}=1-\mathcal{P}_{\text{out}}^{(n)} (0)
=1-F_{\Gamma_{\text{R}}^{(n)}}(1). 
\end{align}

To incorporate the backhaul link reliability factor  in the end-to-end SNR distribution of links from the AP to $\text{D}$ or $\text{E}$ via $\text{S}_n$, we utilize a mixture distribution of the  Bernoulli and exponential RVs. A convex combination of the Bernoulli and exponential distribution is used for the mixture distribution to model the end-to-end $\text{AP}$-$\text{S}_n$-$\text{X}$ link SNR, for $\text{X}\in \{\text{D,E}\}$,  ${\hat{\Gamma}^{(n)}_{\text{SX}}}$ is expressed as 
\cite{Kundu_TVT19}
% \textcolor{red}{define perfect and unreliable in introduction}
\begin{align}\label{mixture_distribution}
f_{\hat{\Gamma}^{(n)}_{\text{SX}}}(x)=(1-s)\delta(x)+s f_{\Gamma^{(n)}_{\text{SX}}}(x),     
\end{align}
where 
%where $\hat{\Gamma}^{(n)}_{SR}(x)$ denotes the end-to-end SNR, 
%$\hat{\Gamma}^{(n)}_{SR}$ of the  AP-$\text{S}_n$-$R$ link, 
 $\delta(x)$ is the Dirac delta function,  ${\Gamma^{(n)}_{\text{SX}}}$ is the SNR of the $\text{S}_n$-$\text{X}$ link when the corresponding backhaul link is active.  Here we note that the factor $(1-s)$ associated with $\delta(x)$ represents the probability that the backhaul is inactive. When $s=1$, it suggests that the backhaul link is active. 
% Similarly, the mixture distribution of end-to-end $AP$-$\text{S}_n$-$\text{E}$ link SNR can be expressed by replacing $\text{D}$ with $E$ in (\ref{mixture_distribution}). 
It is noted from (\ref{mixture_distribution}) that the application of the mixture distribution generalizes the secrecy performance analysis for the transmitter selection schemes under both the perfect backhaul (i.e., $s=1$) and unreliable backhaul (i.e., $s<1$) conditions.
As the backhaul link is common between any end-to-end link leading to $\text{D}$ and $\text{E}$, the backhaul link reliability factor  has to be incorporated in either $\hat{\Gamma}^{(n)}_{\text{SD}}$ or $\hat{\Gamma}^{(n)}_{\text{SE}}$ during the performance analysis (NZSR, SOP, and ESR) of transmitter selection schemes but not in both as depicted in (\ref{mixture_distribution}).  This is because, for a given selected transmitter $n^*$ on the basis of $\hat{\Gamma}^{(n^*)}_{\text{SX}}$, the corresponding $\text{S}_{n^*}$-$\text{X}$ link must be considered. Otherwise, $\hat{\Gamma}^{(n^*)}_{\text{SD}}$ and $\hat{\Gamma}^{(n^*)}_{\text{SE}}$ would become independent.

% This is because, for a given selected transmitter $n^*$ on the basis of $\hat{\Gamma}^{(n^*)}_{\text{SD}}$ (or $\hat{\Gamma}^{(n^*)}_{\text{SE}}$), the link $S_{n^*}$-$\text{D}$ (or $S_{n^*}$-$E$) must to be considered. Otherwise, $\hat{\Gamma}^{(n^*)}_{\text{SD}}$ and $\hat{\Gamma}^{(n^*)}_{\text{SE}}$ would become independent.

In the following sections, we evaluate the NZSR, SOP, and ESR for each transmitter selection scheme under two backhaul link activity knowledge cases, BKU and BKA, with the help of the distribution of ${\Gamma_{\text{R}}^{(n^*)}}$ for a given selected transmitter. The distribution of ${\Gamma_{\text{R}}^{(n^*)}}$ changes according to the transmitter selection schemes and backhaul link activity knowledge, which will be obtained in the next section.  Further, the asymptotic analysis is also carried out in each case to obtain better insights.

The distribution of ${\Gamma_{\text{R}}^{(n^*)}}$  helps us to obtain any secrecy performance (NZSR, SOP, ESR) of both the STS and OTS schemes in a unified manner. It is also noted here that the evaluation of the ESR for the OTS scheme would not have been possible otherwise.
% To evaluate these metrics, we will adapt the distributions used in the above equations according to a particular selection scheme and available backhaul knowledge activity. 

% The distribution of ${\Gamma_{\text{R}}^{(n^*)}}$ for a given selected transmitter changes according to the transmitter selection schemes. This will be obtained in the next section. 

% \begin{enumerate}
%     \item 
% % \textcolor{red}{Use either transmitter or source selection uniformly not both throughout the paper. Change it throughout. Define two backhaul cases with some acronym do not keep it in text.}

%    \item 
% % \textcolor{red}{ Define two backhaul cases with some acronym do not keep it in text. Replace everywhere these two used.}
%    \item 
% % \textcolor{red}{'write backhaul reliability factor' instead of only backhaul reliability everywhere}
% \end{enumerate}
\section{Distribution of $\Gamma_{\text{R}}^{(n^*)}$}\label{sec_distribution_snr_ratio}
% \textcolor{red}{all sections must be with selection} 
In this section, the CDF of $\Gamma_{\text{R}}^{(n^*)}$ for each transmitter selection scheme is evaluated for both BKU and BKA cases. This will then be used to find the NZSR, SOP, and ESR of those transmitter selection schemes. To evaluate the distribution of $\Gamma_{\text{R}}^{(n^*)}$, in the BKU case, the SNR distribution corresponding to the selected transmitter is evaluated first and then the backhaul reliability factor is included in the SNR distribution. This is due to the fact that which backhaul links are available is not known. However, in the BKA case,  the backhaul reliability factor is incorporated in the individual link SNRs first and then the SNR distribution of the selected transmitter is evaluated. This will be clear in the following subsections. 

% \begin{align}
% \mathcal{P}^{(n)}_{\text{NZ}}&=1-  \int_{0}^{\infty} F_{{\Gamma}^{(n)}_{\text{SD}}}(x) f_{\hat{\Gamma}^{(n)}_{\text{SE}}}(x)dx
% \label{non_zero_mines}
% \\
% &=1-  \int_{0}^{\infty} F_{\hat{\Gamma}^{(n)}_{\text{SD}}}(x) f_{{\Gamma}^{(n)}_{\text{SE}}}(x)dx.
% \label{non_zero}
% \end{align}
% The distributions of ${\Gamma}^{(n)}_{SD}$ and ${\Gamma}^{(n)}_{SE}$ corresponding to each transmitter scheme in (\ref{non_zero_mines}) or (\ref{non_zero})
% We need to find the distributions in (\ref{non_zero_mines}) or (\ref{non_zero}) corresponding to each transmitter scheme in order to evaluate NZSR in a closed form.
\subsection{Minimal eavesdropping selection in the BKU case (MIN-ES-BKU)}\label{sub_sec_min_es_bku}
In this section, the transmitter with minimum channel power gain among $\text{S}_n$-$\text{E}$ links where $n\in\{1,\ldots,N\}$ is selected without considering which backhaul links are active as no knowledge of backhaul link activity is available. Hence, the eavesdropping SNR corresponding to the selected transmitter if the backhaul link is active becomes, 
\begin{align}\label{MIN-ES_k} 
\Gamma_{\text{SE}}^{(n^*)} = \min_{n\in \{1,2\ldots,N\}} \{\Gamma_{\text{SE}}^{(n)}\}.
\end{align}
The CDF of $\Gamma_{\text{R}}^{(n^*)}$ under unreliable backhaul conditions including backhaul link reliability factor  is then evaluated 
following (\ref{gamma_r}) as 
% (\ref{gamma_r_min_without})
\begin{align}\label{gamma_r_min_without_0}
F_{\Gamma_{\text{R}}^{(n^*)}}(y)
&= \int_{0}^{\infty} F_{{\Gamma}^{(n)}_{\text{SD}}}(y(x+1)-1) f_ {\hat{\Gamma}_{\text{SE}}^{(n^*)}}(x) dx.
% \nn\\
% &=1-\int_{0}^{\infty}F_{{\Gamma}_{\text{SD}}^{(n)}}(x)\Big((1-s)\delta(x)+sf_{{\Gamma}_{\text{SE}}^{(n^*)}}(x)\Big)dx\nn\\
% &=1-\int_{0}^{\infty}\Big(F_{{\Gamma}_{\text{SD}}^{(n)}}(x)(1-s)\delta(x)+sF_{{\Gamma}_{\text{SD}}^{(n)}}(x)f_{{\Gamma}_{\text{SE}}^{(n^*)}}(x)\Big)dx.
\end{align}
% where $f_{\hat{\Gamma}_{\text{SE}}^{(n^*)}}(x)$ is the PDF of $\hat{\Gamma}_{\text{SE}}^{(n^*)}$ of the end-to-end $AP$-$S_{n^*}$-$E$ link including backhaul link reliability factor . 
Due to the common backhaul link between destination and eavesdropping channel, the CDF of $\Gamma_{\text{SD}}^{(n)}$ conditioned on the selected transmitter is $\hat{\Gamma}_{\text{SD}}^{(n)} =\Gamma_{\text{SD}}^{(n)}$. Thus, the backhaul link reliability factor  is considered in $\hat{\Gamma}_{\text{SE}}^{(n^*)}$ but not in $\hat{\Gamma}_{\text{SD}}^{(n)}$ while evaluating (\ref{gamma_r_min_without_0}).
Following the mixture distribution method shown in (\ref{mixture_distribution}), $f_{\hat{\Gamma}_{\text{SE}}^{(n^*)}}(x)$ is evaluated as
\begin{align}
\label{PDF_SE_MIN-ES_without}
f_{\hat{\Gamma}_{\text{SE}}^{(n^*)}}(x)&=(1-s)\delta(x)+sf_{{\Gamma}_{\text{SE}}^{(n^*)}}(x),
\end{align}
where $f_{{\Gamma}_{\text{SE}}^{(n^*)}}(x)$ is provided in (\ref{PDF_SE_MIN-ES}) in \textit{Lemma \ref{lemma1}}. 
Finally, using (\ref{PDF_SE_MIN-ES}) and $F_{\Gamma_{\text{SD}}^{(n)}}(x)$ from (\ref{cdf_SD}), $F_{\Gamma_{\text{R}}^{(n^*)}}(y)$ is evaluated as 
\begin{align}
\label{eq_integral_terms}
&F_{\Gamma_{\text{R}}^{(n^*)}}(y) =(1-s)\int_{0}^{\infty}\delta(x)dx\nn\\
&+s\int_{0}^{\infty}\Big(1-e^{-\lambda_{\text{D}}(y(x+1)-1)}\Big)\sum_{\mathbf{i}\in\mathcal{M}^{(N)}}\binom{N}{i_1,\ldots,i_K}\nn\\
&\times\Bigg(\prod_{k=1}^{K}\Bigg(\frac{\prod_{m=1}^{K}\lambda_{\text{E}}^{(m)}}{\lambda_{\text{E}}^{(k)}\prod\limits_{\substack{j=1\\j\ne k}}^K(\lambda_{\text{E}}^{(j)}-\lambda_{\text{E}}^{(k)})}\Bigg)^{i_k}\Bigg)\tilde{\lambda}_{\text{E}}^{(k)} e^{-\tilde{\lambda}_{\text{E}}^{(k)}x}dx \\
\label{gamma_r_min_without}
&=1-\sum_{\mathbf{i}\in\mathcal{M}^{(N)}}\binom{N}{i_1,\ldots,i_K}\nn\\
&\times\Bigg(\prod_{k=1}^{K}\Bigg(\frac{\prod_{m=1}^{K}\lambda_{\text{E}}^{(m)}}{\lambda_{\text{E}}^{(k)}\prod\limits_{\substack{j=1\\j\ne k}}^K(\lambda_{\text{E}}^{(j)}-\lambda_{\text{E}}^{(k)})}\Bigg)^{i_k}\Bigg)\frac{s\tilde{\lambda}_{\text{E}}^{(k)}e^{-\lambda_{\text{D}} (y-1)}}{(y\lambda_{\text{D}} +\tilde{\lambda}_{\text{E}}^{(k)})},
\end{align}
where   $\mathcal{M}^{(N)}$, $\binom{N}{i_1,\ldots,i_K}$,  $i_{k}$ for each $k\in\{1 ,\ldots, K\}$, and $\tilde{\lambda}_{\text{E}}^{(k)}$ are defined in Lemma \ref{lemma1}.
The first term in (\ref{eq_integral_terms}) represents the case where the backhaul link is inactive; therefore, $F_{{\Gamma}^{(n)}_{\text{SD}}}(x)=1$ for any value of $x$. Whereas the second term represents the case where the backhaul link is active, and hence, $F_{{\Gamma}^{(n)}_{\text{SD}}}(x)$ in the second term follows exponential CDF as in (\ref{pdf_SD}). 

\textit{Remark}: The NZSR  and SOP of the system are evaluated directly from $F_{\Gamma_{\text{R}}^{(n^*)}}(y)$ in (\ref{gamma_r_min_without}) following (\ref{eq_NZSR_from_SOP}) and (\ref{eq_new_SOP_from_GammaR}), respectively.

\begin{lemma}\label{lemma1}
 The PDF of ${{\Gamma}_{\text{SE}}^{(n^*)}}(x)$ is given as 
\begin{align}
\label{PDF_SE_MIN-ES}
&f_{{\Gamma}_{\text{SE}}^{(n^*)}}(x)=\sum_{\mathbf{i}\in\mathcal{M}^{(N)}}\binom{N}{i_1,\ldots,i_K}\nn\\
&\times\Bigg(\prod_{k=1}^{K}\Bigg(\frac{\prod_{m=1}^{K}\lambda_{\text{E}}^{(m)}}{\lambda_{\text{E}}^{(k)}\prod\limits_{\substack{j=1\\j\ne k}}^K(\lambda_{\text{E}}^{(j)}-\lambda_{\text{E}}^{(k)})}\Bigg)^{i_k}\Bigg)\tilde{\lambda}_{\text{E}}^{(k)} e^{-\tilde{\lambda}_{\text{E}}^{(k)}x},
\end{align}
where   $\mathcal{M}^{(\omega)}$ is the set of integer vectors $[i_{1},\ldots,i_{K}]$ containing $K$ elements such that $i_{k}\in\{0 ,\ldots, \omega\}$ for each $k\in\{1 ,\ldots, K\}$ and $\sum_{k=1}^{K}i_k=\omega$ where $\omega$ is any whole number, i.e., $\omega\in\{0,1,2,\ldots\}$, $\binom{N}{i_1,\ldots,i_K}=\frac{N!}{i_1!i_2!\ldots i_K}$ is a multinomial coefficient, and $\tilde{\lambda}_{\text{E}}^{(k)}=\sum_{k=1}^{K}i_k  \lambda_{\text{E}}^{(k)} $.
\end{lemma}

\begin{proof}
% \textit{Proof:} 
The PDF $f_{{\Gamma}_{\text{SE}}^{(n^*)}}(x)$ is derived from the CDF of $\Gamma_{\text{SE}}^{(n^*)}$. The CDF $F_{{\Gamma}_{\text{SE}}^{(n^*)}}(x)$ is derived using the CDF of $\Gamma_{\text{SE}}^{(n)}(x)$ in (\ref{CDF_SE}) and utilizing the multinomial theorem as
\begin{align}\label{cdf_se_minimunm}
&F_{{\Gamma}_{\text{SE}}^{(n^*)}}(x)=\mathbb{P}\Big[\min_{{n \in \{1,2,\ldots, N\}}}\{\Gamma_{\text{SE}}^{(n)}\}\le x\Big]\nn\\
&=1-\Big(1-\mathbb{P}\Big[\Gamma_{\text{SE}}^{(n)}\le x\Big]\Big)^N\nn\\
% \nn \\
% &=1-\prod_{n=1}^{N}\Big(1-\mathbb{P}\left[\Gamma_{\text{SE}}^{(n)}\le x\right]\Big) \nn\\
% &
% &=\mathbb{P}\Big[\min_{{n \in \{1,2,\ldots, N\}}}\{\Gamma_{\text{SE}}^{(n)}\}\le x\Big]
&=1-\Bigg(\sum_{k=1}^{K}\frac{\Big(\prod_{m=1}^{K}\lambda_{\text{E}}^{(m)}\Big)e^{-\lambda_{\text{E}}^{(k)} x}}{\lambda_{\text{E}}^{(k)}\prod\limits_{\substack{j=1\\j\ne k}}^K(\lambda_{\text{E}}^{(j)}-\lambda_{\text{E}}^{(k)})}\Bigg)^{N} \nn \\
&=1-\sum_{{\mathbf{i}\in\mathcal{M}^{(N)}}}\binom{N}{i_1,\ldots,i_K}\nn\\
&\times\Bigg(\prod_{k=1}^{K}\Bigg(\frac{\prod_{m=1}^{K}\lambda_{\text{E}}^{(m)}}{\lambda_{\text{E}}^{(k)}\prod\limits_{\substack{j=1\\j\ne k}}^K(\lambda_{\text{E}}^{(j)}-\lambda_{\text{E}}^{(k)})}\Bigg)^{i_k}\Bigg)e^{-\tilde{\lambda}_{\text{E}}^{(k)}x},
\end{align}
By differentiating (\ref{cdf_se_minimunm}), the PDF $f_{{\Gamma}_{\text{SE}}^{(n^*)}}(x)$ is achieved as shown in (\ref{PDF_SE_MIN-ES}). 
\end{proof}

\subsection{Minimal eavesdropping selection in the BKA case (MIN-ES-BKA)}\label{sub_sec_min_es_bka}
% \textcolor{red}{correct as per without backhaul}
In this section, the transmitter with minimum channel power gain among $\text{S} _n$-$\text{E}$ links is selected among all links for which backhaul links are active. Assuming $\mathcal{S} \subseteq \{1,2,\ldots, N\}$ is the subset of transmitters for which the backhaul links are active, the end-to-end eavesdropping SNR including backhaul link reliability factor  is 
\begin{align} \label{eq_gain_min_with}
\hat{\Gamma}_{\text{SE}}^{(n^* )}=\min_{n\in\mathcal{S}}\{\hat{\Gamma}_{\text{SE}}^{(n)}\}.
% = \min_{{n \in \{1,2\ldots, N\}}}|{h}_{S_n D}^{'}|^2.
\end{align} 
Therefore, $F_{\Gamma_{\text{R}}^{(n^*)}}(y)$ is evaluated following (\ref{gamma_r_min_without_0}) 
% as
% \begin{align}\label{non_zero_min_es_with}
% &F_{\Gamma_{\text{R}}^{(n^*)}}(y)=
% \int_{0}^{\infty} F_{{\Gamma}^{(n)}_{\text{SD}}}(y(x+1)-1) f_ {\hat{\Gamma}_{\text{SE}}^{(n^*)}}(x) dx
% \end{align}
where $f_{\hat{\Gamma}_{\text{SE}}^{(n^*)}}(x)$ is given in (\ref{eq_gain_min_with}).
To evaluate the $F_{\Gamma_{\text{R}}^{(n^*)}}(y)$, the distribution of $\hat{\Gamma}_{\text{SE}}^{(n^* )}$ is required as in the previous section and is evaluated in (\ref{PDF_SE_MINES_with}) in \textit{ Lemma \ref{lemma2}}. With the help of $f_{\hat{\Gamma}_{\text{SE}}^{(n^*)}}(x)$ from (\ref{PDF_SE_MINES_with}) and following the similar steps that were carried out for (\ref{gamma_r_min_without}), $F_{\Gamma_{\text{R}}^{(n^*)}}(y)$ is evaluated as
\begin{align}\label{gamma_r_min_with}
&F_{\Gamma_{\text{R}}^{(n^*)}}(y)=
(1-s)^N\int_{0}^{\infty}\delta(x)dx\nn\\
&+\int_{0}^{\infty}\Big(1-e^{-\lambda_{\text{D}}(\rho x+y-1)}\Big)\sum_{n=1}^{N}\sum_{\mathbf{i}\in\mathcal{M}^{(n)}}\binom{N}{n}(1-s)^{N-n}\nn\\
&\times s^n\binom{n}{i_1,\ldots,i_K}\Bigg(\prod_{k=1}^{K}\Bigg(\frac{\prod_{m=1}^{K}\lambda_{\text{E}}^{(m)}}{\lambda_{\text{E}}^{(k)}\prod\limits_{\substack{j=1\\j\ne k}}^K(\lambda_{\text{E}}^{(j)}-\lambda_{\text{E}}^{(k)})}\Bigg)^{i_k}\Bigg)\nn\\
&\times\tilde{\lambda}_{\text{E}}^{(k)}e^{-\tilde{\lambda}_{\text{E}}^{(k)}x}dx \nn \\
&=1-\sum_{n=1}^{N}\sum_{\mathbf{i}\in\mathcal{M}^{(n)}}\binom{N}{n}\binom{n}{i_1,\ldots,i_K}(1-s)^{N-n}s^n\nn\\
&\times\Bigg(\prod_{k=1}^{K}\Bigg(\frac{\prod_{m=1}^{K}\lambda_{\text{E}}^{(m)}}{\lambda_{\text{E}}^{(k)}\prod\limits_{\substack{j=1\\j\ne k}}^K(\lambda_{\text{E}}^{(j)}-\lambda_{\text{E}}^{(k)})}\Bigg)^{i_k}\Bigg)\frac{\tilde{\lambda}_{\text{E}}^{(k)}e^{-\lambda_{\text{D}}(y-1)}}{y\lambda_{\text{D}}+\tilde{\lambda}_{\text{E}}^{(k)}}.
\end{align}

\textit{Remark}: The NZSR and SOP can be evaluated directly from $F_{\Gamma_{\text{R}}^{(n^*)}}(y)$ in (\ref{gamma_r_min_with}) following (\ref{eq_NZSR_from_SOP}) and (\ref{eq_new_SOP_from_GammaR}), respectively.

\begin{lemma}\label{lemma2}
The distribution of $\hat{\Gamma}_{\text{SE}}^{(n^* )}$ is  
\begin{align}\label{PDF_SE_MINES_with}
&f_{\hat{\Gamma}_{\text{SE}}^{(n^*)}}(x)=(1-s)^N\delta(x)+\sum_{n=1}^{N}\sum_{\mathbf{i}\in\mathcal{M}^{(n)}}\binom{N}{n}\binom{n}{i_1,\ldots,i_K}\nn\\
&\times (1-s)^{N-n}s^n\Bigg(\prod_{k=1}^{K}\Bigg(\frac{\prod_{m=1}^{K}\lambda_{\text{E}}^{(m)}}{\lambda_{\text{E}}^{(k)}\prod\limits_{\substack{j=1\\j\ne k}}^K(\lambda_{\text{E}}^{(j)}-\lambda_{\text{E}}^{(k)})}\Bigg)^{i_k}\Bigg)\nn\\
&\times\tilde{\lambda}_{\text{E}}^{(k)}e^{-\tilde{\lambda}_{\text{E}}^{(k)}x}.
\end{align}
\end{lemma}
\begin{proof}
% \textit{Proof:}
The PDF $f_{\hat{\Gamma}_{\text{SE}}^{(n^*)}}(x)$ is derived from the CDF of $\hat{\Gamma}_{\text{SE}}^{(n^*)}$.
To find the CDF of $\hat{\Gamma}_{\text{SE}}^{(n^*)}$, we consider evaluating CDF in two mutually exclusive events wherein i) all the backhaul links are inactive and ii) at least one of the links is active. In the first event, as the probability that all links are inactive is $(1-s)^N$, the CDF in this event is 
\begin{align}\label{CDF_SE_MIN-ES_with_mix1}
F_{\hat{\Gamma}_{\text{SE}}^{(n^*)}}(x)&=\big(\mathbb{P}\left[\mathbb{I}_n=0\right]\big)^N u(x)= (1-s)^N u(x),
\end{align}
where $u(x)$ is the unit step function. 
In the second event, the CDF is directly obtained as 
\begin{align}
\label{CDF_SE_MIN-ES_with_mix2}
&F_{\hat{\Gamma}_{\text{SE}}^{(n^*)}}(x)=\sum_{n=1}^{N}\binom{N}{n}(1-s)^{N-n}s^n\Big(\mathbb{P}\Big[\min_{n\in\mathcal{S}}\{{\Gamma}_{\text{SE}}^{(n)}\}\le x\Big]\Big) \nn\\
&=\sum_{n=1}^{N}\binom{N}{n}(1-s)^{N-n}s^n\Big(1-\Big(1-\mathbb{P}\Big[{\Gamma}_{\text{SE}}^{(n)}\le x\Big]\Big)^n\Big) \nn\\
% &=\sum_{n=1}^{N}\binom{N}{n}(1-s)^{N-n}s^n\lb(1-\lb(1-1+\sum_{k=1}^{K}\frac{\Big(\prod_{m=1}^{K}\lambda_{\text{E}}^{(m)}\Big) e^{-\lambda_{\text{E}}^{(k)} x}}{\lambda_{\text{E}}^{(k)}\prod\limits_{\substack{j=1\\j\ne k}}^K(\lambda_{\text{E}}^{(j)}-\lambda_{\text{E}}^{(k)})}\rb)^n\rb)\nn\\
&=\sum_{n=1}^{N}\binom{N}{n}(1-s)^{N-n}s^n\nn\\
&\times\Bigg(1-\Bigg(\sum_{k=1}^{K}\frac{\Big(\prod_{m=1}^{K}\lambda_{\text{E}}^{(m)}\Big) e^{-\lambda_{\text{E}}^{(k)} x}}{\lambda_{\text{E}}^{(k)}\prod\limits_{\substack{j=1\\j\ne k}}^K(\lambda_{\text{E}}^{(j)}-\lambda_{\text{E}}^{(k)})}\Bigg)^n\Bigg).
\end{align}
Finally, the CDF  $F_{\hat{\Gamma}_{\text{SE}}^{(n^*)}}(x)$ is obtained by adding (\ref{CDF_SE_MIN-ES_with_mix1}) and (\ref{CDF_SE_MIN-ES_with_mix2})  as both the cases arise independently. Similar to (\ref{cdf_se_minimunm}), we apply the multinomial theorem in (\ref{CDF_SE_MIN-ES_with_mix2}) as well. The PDF $f_{\hat{\Gamma}_{\text{SE}}^{(n^*)}}(x)$ in (\ref{PDF_SE_MINES_with}) is obtained by differentiating (\ref{CDF_SE_MIN-ES_with_mix2}).  
\end{proof}

We here note that it is challenging to find the minimum among RVs in (\ref{eq_gain_min_with}) for which backhaul links are active as these RVs have a mixture distribution of Bernoulli and exponential. As Bernoulli RV always takes the minimum value of zero, one might end up getting zero in (\ref{eq_gain_min_with}) if proper modeling is not implemented for active backhaul links.  We note here that the procedure adopted in Lemma \ref{lemma2} for incorporating the backhaul reliability factor in the eavesdropping link SNR and finding the minimum of RVs which includes mixture distribution is novel.

 \subsection{Traditional transmitter selection in the BKU case (TTS-BKU) }\label{sub_sec_tts_bku}In this section, the transmitter is selected based on the maximum channel power gain among the $\text{S}_n$-$\text{D}$ links without knowing which backhaul links are active. In this case, the selected transmitter provides the following destination SNR if the corresponding backhaul link is active 
% When the knowledge of backhaul link activity is not available, the transmitter for which the channel power gain is maximum will be selected independent of the backhaul link reliability factor  of backhaul as 
\begin{align}\label{eq_gain_max_without}
{\Gamma}_{\text{SD}}^{(n^*)}=  \max_{n\in \{1,2\ldots, N\}}\{ {\Gamma}_{\text{SD}}^{(n)}\}.
\end{align}
To derive the  CDF of $\Gamma_{\text{R}}^{(n^*)}$, the backhaul link reliability factor  is included in $\hat{\Gamma}_{\text{SD}}^{(n)}$ and hence, the  CDF of $\Gamma_{\text{R}}^{(n^*)}$ is evaluated following (\ref{gamma_r}) as
\begin{align}\label{non_zero_tts}
F_{\Gamma_{\text{R}}^{(n^*)}}(y)&=\int_{0}^{\infty} F_{\hat{\Gamma}^{(n^*)}_{\text{SD}}}(y(x+1)-1) f_ {{\Gamma}_{\text{SE}}^{(n)}}(x) dx.   
\end{align}
% where $F_{\hat{\Gamma}^{(n^*)}_{\text{SD}}}(x)$  is the CDF of the end-to-end SNR ${\hat{\Gamma}^{(n^*)}_{\text{SD}}}$ of the selected $AP$-$S_{n^*}$-$\text{D}$ link. 
The CDF $F_{\hat{\Gamma}^{(n^*)}_{\text{SD}}}(x)$ above is expressed following the mixture distribution from (\ref{mixture_distribution}) as
\begin{align}\label{CDF_sd_TTS_without}
F_{\hat{\Gamma}^{(n^*)}_{\text{SD}}}(x)&=(1-s)u(x)+sF_{\Gamma^{(n^*)}_{\text{SD}}}(x),
% \nn\\
% &=(1-s)u(x)+s\Big(1-\sum_{n=1}^{N}\binom{N}{n}(-1)^{n+1}e^{-n\lambda_{\text{D}} x}\Big).
\end{align} 
where $F_{{\Gamma}^{(n^*)}_{\text{SD}}}(x)$ is derived in (\ref{CDF_TTS_without_max}) in Lemma \ref{lemma3}. 
Next,  $F_{\Gamma_{\text{R}}^{(n^*)}}(y)$ in (\ref{non_zero_tts}) is evaluated using $F_{\hat{\Gamma}_{\text{SD}}^{(n^*)}}(x)$ from (\ref{CDF_sd_TTS_without}) and 
 $f_{{\Gamma}_{\text{SE}}^{(n)}}(x)$ from (\ref{PDF_SE}) as
\begin{align}\label{gamma_r_tts_without}
&F_{\Gamma_{\text{R}}^{(n^*)}}(y)=\int_{0}^{\infty}\Big((1-s)u(x)+s\Big(1-\sum_{n=1}^{N}\binom{N}{n}(-1)^{n+1}\nn\\
&\times e^{-n\lambda_{\text{D}}(y(x+1)-1)}\Big)\Big)\sum_{k=1}^{K}\frac{\prod_{m=1}^{K}\lambda_{\text{E}}^{(m)}}{\prod\limits_{\substack{j=1\\j\ne k}}^K(\lambda_{\text{E}}^{(j)}-\lambda_{\text{E}}^{(k)})}e^{-\lambda_{\text{E}}^{(k)} x}dx\nn\\
&=1-\sum_{n=1}^{N}\sum_{k=1}^{K}\binom{N}{n}\frac{(-1)^{n+1}s\Big(\prod_{m=1}^{K}\lambda_{\text{E}}^{(m)}\Big)e^{-\lambda_{\text{D}} (y-1)}}{(ny\lambda_{\text{D}}+ \lambda_{\text{E}}^{(k)})\prod\limits_{\substack{j=1\\j\ne k}}^K(\lambda_{\text{E}}^{(j)}-\lambda_{\text{E}}^{(k)})}.
\end{align}

\textit{Remark}: The NZSR and SOP can be evaluated directly from $F_{\Gamma_{\text{R}}^{(n^*)}}(y)$ in (\ref{gamma_r_tts_without}) following (\ref{eq_NZSR_from_SOP}) and (\ref{eq_new_SOP_from_GammaR}), respectively.

\begin{lemma}\label{lemma3}
The CDF of ${\Gamma^{(n^*)}_{\text{SD}}}(x)$ without including the backhaul link reliability factor is expressed as
\begin{align}\label{CDF_TTS_without_max}
F_{\Gamma^{(n^*)}_{\text{SD}}}(x)&=1-\sum_{n=1}^{N}\binom{N}{n}(-1)^{n+1}e^{-n\lambda_{\text{D}} x}.
\end{align} 
\end{lemma}
\begin{proof}
% \textit{Proof:} 
$F_{\Gamma^{(n^*)}_{\text{SD}}}(x)$ is evaluated using the definition of CDF  as
\begin{align}\label{lemma3proof}
F_{\Gamma^{(n^*)}_{\text{SD}}}(x)&=\mathbb{P}\Big[\max_{{n \in \{1,2,\ldots, N\}}}\{{\Gamma}_{\text{SD}}^{(n)}\}\le x\Big]\nn\\
&=\prod_{n=1}^{N}\mathbb{P}\Big[\Gamma_{\text{SD}}^{(n)}\}\le x\Big].
% =1-\sum_{n=1}^{N}\binom{N}{n}(-1)^{n+1}e^{-n\lambda_{\text{D}} x}.
\end{align}
Using (\ref{cdf_SD}), (\ref{CDF_TTS_without_max}) provides  (\ref{lemma3proof}). 
\end{proof}

\subsection{Traditional transmitter selection in the BKA case (TTS-BKA)}\label{sub_sec_tts_bka}
In this section, the transmitter is selected based on the maximum channel power gain among the $\text{S}_n$-$\text{D}$ links whose backhaul links are active. In this case, the destination SNR including the backhaul link reliability factor becomes
\begin{align} 
\label{eq_gain_max_with}
\hat{\Gamma}_{\text{SD}}^{(n^*)}=\max_{n\in\mathcal{S}}\{\hat{\Gamma}_{\text{SD}}^{(n)}\}.
\end{align}
Therefore, the  CDF of $\Gamma_{\text{R}}^{(n^*)}$ is evaluated following (\ref{non_zero_0}) as
\begin{align}\label{non_zero_tts_with}
&F_{\Gamma_{\text{R}}^{(n^*)}}(y)=
\int_{0}^{\infty} F_{\hat{\Gamma}_{\text{SD}}^{(n^*)}}(y(x+1)-1)f_{{\Gamma}_{\text{SE}}^{(n)}}(x)dx,
\end{align}
where $F_{\hat{\Gamma}_{\text{SD}}^{(n^*)}}(x)$ is evaluated in (\ref{CDF_SD_TTS_with}) in Lemma \ref{lemma4}.
Next, $F_{\Gamma_{\text{R}}^{(n^*)}}(y)$ is obtained using $F_{\hat{\Gamma}_{\text{SD}}^{(n^*)}}(x)$ from (\ref{CDF_SD_TTS_with}) and $f_{{\Gamma}_{\text{SE}}^{(n)}}(x)$ from (\ref{PDF_SE}) in (\ref{non_zero_tts_with}) as
\begin{align}\label{gamma_r_tts_with}
&F_{\Gamma_{\text{R}}^{(n^*)}}(y)
% &=\int_{0}^{\infty} F_{\hat{\Gamma}^{(n^*)}_{\text{SD}}}(\rho(x+1)-1) f_ {{\Gamma}_{\text{SE}}^{(n)}}(x) dx\nn\\
=\int_{0}^{\infty}\Big(1-\sum_{n=1}^{N}\binom{N}{n}(1-s)^{N-n}u(x)s^{n}\sum_{q=1}^{n}\nn\\
&\binom{n}{q}(-1)^{q+1}e^{-q\lambda_{\text{D}} (y(x+1)-1)}\Big)\sum_{k=1}^{K}\frac{\prod_{m=1}^{K}\lambda_{\text{E}}^{(m)}e^{-\lambda_{\text{E}}^{(k)} x}dx }{\prod\limits_{\substack{j=1\\j\ne k}}^K(\lambda_{\text{E}}^{(j)}-\lambda_{\text{E}}^{(k)})} \nn \\
&=1-\sum_{n=1}^{N}\sum_{k=1}^{K}\sum_{q=1}^{n}\binom{N}{n}\binom{n}{q}(1-s)^{N-n}s^{n}(-1)^{q+1}\nn\\
&\times\frac{\Big(\prod_{m=1}^{K}\lambda_{\text{E}}^{(m)}\Big)e^{-\lambda_{\text{D}} (y-1)}}{(qy\lambda_{\text{D}} +\lambda_{E})\prod\limits_{\substack{j=1\\j\ne k}}^K(\lambda_{\text{E}}^{(j)}-\lambda_{\text{E}}^{(k)})}.
\end{align}
% As independent identically distributed Bernoulli RVs govern the activity status of backhaul links,

\textit{Remark}: The NZSR and the SOP can be evaluated directly from $F_{\Gamma_{\text{R}}^{(n^*)}}(y)$ in (\ref{gamma_r_tts_with}) following (\ref{eq_NZSR_from_SOP}) and (\ref{eq_new_SOP_from_GammaR}), respectively.

\begin{lemma}\label{lemma4}
The CDF of the SNR $\hat{\Gamma}_{\text{SD}}^{(n^*)}$ when the backhaul link activity knowledge is available, is 
\begin{align}\label{CDF_SD_TTS_with}
F_{\hat{\Gamma}^{(n^*)}_{\text{SD}}}(x)
&=1-\sum_{n=1}^{N}\sum_{q=1}^{n}\binom{N}{n}\binom{n}{q}(1-s)^{N-n}s^{n}(-1)^{q+1}\nn\\
&\times e^{-q\lambda_{\text{D}} x}.
\end{align}
\end{lemma}
\begin{proof}
The CDF $F_{\hat{\Gamma}_{\text{SD}}^{(n^*)}}$ is  evaluated as
\begin{align}\label{lemma4proof}
F_{\hat{\Gamma}^{(n^*)}_{\text{SD}}}(x)
&=\mathbb{P}\Big[\max_{{n \in \{1,2\ldots, N\}}}\{\hat{\Gamma}_{\text{SD}}^{(n)}\}\le x\Big]\nn\\
&=\prod_{n=1}^{N}\Big(\mathbb{P}\Big[\hat{\Gamma}_{\text{SD}}^{(n)}\le x\Big] \Big)
\nn \\
&=\Big((1-s)u(x)+s(1-e^{-\lambda_{\text{D}} x})\Big)^N.
% &=\sum_{n=0}^{N}\binom{N}{n}(1-s)^{N-n}u(x)s^{n}\Big(1-e^{-\lambda_{\text{D}} x}\Big)^n\nn\\
% &=(1-s)^Nu(x)+\sum_{n=1}^{N}\binom{N}{n}(1-s)^{N-n}u(x)s^{n}\Big(1-\sum_{q=1}^{n}\binom{n}{q}(-1)^{q+1}e^{-q\lambda_{\text{D}} x}\Big)\nn\\
% &=(1-s)^Nu(x)+\sum_{n=1}^{N}\binom{N}{n}(1-s)^{N-n}u(x)s^{n}\nn\\
% &-\sum_{n=1}^{N}\binom{N}{n}(1-s)^{N-n}u(x)s^{n}\sum_{q=1}^{n}\binom{n}{q}(-1)^{q+1}e^{-q\lambda_{\text{D}} x}\nn\\
% &=1-\sum_{n=1}^{N}\binom{N}{n}(1-s)^{N-n}s^{n}\sum_{q=1}^{n}\binom{n}{q}(-1)^{q+1}e^{-q\lambda_{\text{D}} x},   
\end{align}
In (\ref{lemma4proof}), we have used the CDF of ${\hat{\Gamma}^{(n)}_{\text{SD}}}(x)$ following the mixture distribution in  (\ref{mixture_distribution}) as 
\begin{align}
\label{CDF_SNR_SD_with}    
F_{\hat{\Gamma}^{(n)}_{\text{SD}}}(x)&=(1-s)u(x)+s(1-e^{-\lambda_{\text{D}} x}).
\end{align}
Finally, (\ref{CDF_SD_TTS_with}) is evaluated from (\ref{lemma4proof}) using two successive binomial expansions.
\end{proof}

% \begin{lemma}\label{lemma4}
% The CDF of the SNR $\hat{\Gamma}_{\text{SD}}^{(n^*)}$ when the backhaul link activity knowledge is available, is written as 
% \begin{align}\label{CDF_SD_TTS_with}
% F_{\hat{\Gamma}^{(n^*)}_{\text{SD}}}(x)
% % &=\mathbb{P}\Big[\max_{{n \in \{1,2\ldots, N\}}}\{\hat{\Gamma}_{\text{SD}}^{(n)}\}\le x\Big]
% =\prod_{n=1}^{N}\Big(\mathbb{P}\Big[\hat{\Gamma}_{\text{SD}}^{(n)}\le x\Big] \Big)
% % \nn \\
% % &=\Big((1-s)u(x)+s(1-e^{-\lambda_{\text{D}} x})\Big)^N \nn \\
% % &=\sum_{n=0}^{N}\binom{N}{n}(1-s)^{N-n}u(x)s^{n}\Big(1-e^{-\lambda_{\text{D}} x}\Big)^n\nn\\
% %&=(1-s)^Nu(x)+\sum_{n=1}^{N}\binom{N}{n}(1-s)^{N-n}u(x)s^{n}\Big(1-\sum_{q=1}^{n}\binom{n}{q}(-1)^{q+1}e^{-q\lambda_{\text{D}} x}\Big)\nn\\
% %&=(1-s)^Nu(x)+\sum_{n=1}^{N}\binom{N}{n}(1-s)^{N-n}u(x)s^{n}\nn\\
% %&-\sum_{n=1}^{N}\binom{N}{n}(1-s)^{N-n}u(x)s^{n}\sum_{q=1}^{n}\binom{n}{q}(-1)^{q+1}e^{-q\lambda_{\text{D}} x}\nn\\
% %&=1-\sum_{n=1}^{N}\binom{N}{n}(1-s)^{N-n}s^{n}\sum_{q=1}^{n}\binom{n}{q}(-1)^{q+1}e^{-q\lambda_{\text{D}} x}\nn\\
% &=1-\sum_{n=1}^{N}\sum_{q=1}^{n}\binom{N}{n}\binom{n}{q}(1-s)^{N-n}s^{n}(-1)^{q+1}e^{-q\lambda_{\text{D}} x}.
% \end{align}
% In (\ref{CDF_SD_TTS_with}), we use the CDF of ${\hat{\Gamma}^{(n)}_{\text{SD}}}(x)$ following the mixture distribution in  (\ref{mixture_distribution}) as 
% \begin{align}
% \label{CDF_SNR_SD_with}    
% F_{\hat{\Gamma}^{(n)}_{\text{SD}}}(x)&=(1-s)u(x)+s(1-e^{-\lambda_{\text{D}} x}).
% \end{align}    
% \end{lemma}

\color{black}
\subsection{Optimal transmitter selection in  the BKU case (OTS-BKU)}\label{sub_sec_ots_bku}
% \textcolor{red}{correct as per MIN-ES}
In this section, the transmitter is selected for which the instantaneous achievable secrecy rate  $C^{(n)}_S$ in (\ref{secrecy_capacity}) among all $n\in\{1,\ldots, N\}$ is maximum  without considering which backhaul links are active as no knowledge
of backhaul link activity is available. This alternatively means that $\Gamma_{\text{R}}^{(n)}$ among all $n\in\{1,\ldots, N\}$ is maximum.
Hence, the selected transmitter is denoted by 
\begin{align}\label{_max_OTS_without}
n^*=\arg \max_{n\in \{1,2\ldots, N\}}\{ {\Gamma_{\text{R}}^{(n)}}\}.
\end{align}
The selected transmitter has uncertainty whether it is active or not, hence, the CDF of $\Gamma_{\text{R}}^{(n^*)}$ is modeled using the mixture distribution method described in (\ref{mixture_distribution}) as
% \color{red}
\begin{align}\label{gamma_r_ots_without}
F_{\Gamma_{\text{R}}^{(n^*)}}(y)&= (1-s)\times1+s \mathbb P\Big[\max_{n\in \{1,2\ldots, N\}}\{\Gamma_{\text{R}}^{(n)}\} \le y\Big].
\end{align}
% \color{black}
In (\ref{gamma_r_ots_without}), the second term in the summation contains the CDF of $\Gamma_{\text{R}}^{(n^*)}$ when the corresponding backhaul link is active, which is evaluated following (\ref{gamma_r}) with the help of $F_{{\Gamma}_{\text{SD}}^{(n)}}(x)$ from (\ref{cdf_SD}) and $f_{{\Gamma}_{\text{SE}}^{(n)}}(x)$ from (\ref{PDF_SE}) as
\begin{align}\label{NZ_OS_without}
&\mathbb P\Big[\max_{n\in \{1,2\ldots, N\}}\{ \Gamma_{\text{R}}^{(n)}\} \le y\Big]\nn\\
&=\Bigg(\int_{0}^{\infty}F_{{\Gamma}_{\text{SD}}^{(n)}}(y(x+1)-1)f_{{\Gamma}_{\text{SE}}^{(n)}}(x)dx \Bigg)^N\nn\\
&=1-\sum_{n=1}^{N}\binom{N}{n}(-1)^{n+1}\nn\\
&\times\Bigg(\sum_{k=1}^{K}\frac{\Big(\prod_{m=1}^{K}\lambda_{\text{E}}^{(m)}\Big)e^{-\lambda_{\text{D}}(y-1)}}{(\lambda_{\text{D}}y+\lambda_{\text{E}}^{(k)})\prod\limits_{\substack{j=1\\j\ne k}}^K(\lambda_{\text{E}}^{(j)}-\lambda_{\text{E}}^{(k)})} \Bigg)^n.
\end{align}
\textit{Remark}: The NZSR and the SOP can be evaluated directly from $F_{\Gamma_{\text{R}}^{(n^*)}}(y)$ in (\ref{gamma_r_ots_without}) following (\ref{eq_NZSR_from_SOP}) and (\ref{eq_new_SOP_from_GammaR}), respectively.

\subsection{Optimal transmitter selection in the BKA case (OTS-BKA)}\label{sub_sec_ots_bka}
The availability of backhaul link activity knowledge permits us to find the maximum $C^{(n)}_S$ among all links with active backhaul links, which is expressed as
\begin{align}\label{max_OTS_with}
n^* =\arg \max_{n\in \mathcal S} \{ \Gamma_{\text{R}}^{(n)}\}.
\end{align}
Consequently, the  CDF of $\Gamma_{\text{R}}^{(n^*)}$ is evaluated following (\ref{gamma_r}) including the backhaul link reliability factor  in $\hat{\Gamma}_{\text{SD}}^{(n)}$ as 
\begin{align}
\label{NZ_OS_with_0}
F_{\Gamma_{\text{R}}^{(n^*)}}(y)
% &=
% \mathbb P\Big[\max_{n\in \mathcal S}\{C^{(n)}_{S}\} \le y\Big]
=\Bigg(\int_{0}^{\infty}F_{\hat{\Gamma}_{\text{SD}}^{(n)}}(y(x+1)-1)f_{{\Gamma}_{\text{SE}}^{(n)}}(x)dx \Bigg)^N.
\end{align}
% We note here that if the probability that the instantaneous secrecy capacity $C_S^{(n)}$ among the active links is less than zero, then the probability that the instantaneous secrecy capacity $C_S^{(n)}$ of all the links is less than zero is the same. 
Therefore, $F_{\Gamma_{\text{R}}^{(n^*)}}(x)$ is evaluated using $F_{\hat{\Gamma}^{(n)}_{\text{SD}}}(x)$ from  (\ref{CDF_SNR_SD_with}) and $f_{{\Gamma}_{\text{SE}}^{(n)}}(x)$ from (\ref{PDF_SE}) in (\ref{NZ_OS_with_0})  as 
\begin{align}\label{gamma_r_ots_with}
F_{\Gamma_{\text{R}}^{(n^*)}}(y)
&=1-\sum_{n=1}^{N}\binom{N}{n}(-1)^{n+1}s^n\nn\\
&\times\Bigg(\sum_{k=1}^{K}\frac{\Big(\prod_{m=1}^{K}\lambda_{\text{E}}^{(m)}\Big)e^{-\lambda_{\text{D}}(y-1)}}{(\lambda_{\text{D}}y+\lambda_{\text{E}}^{(k)})\prod\limits_{\substack{j=1\\j\ne k}}^K(\lambda_{\text{E}}^{(j)}-\lambda_{\text{E}}^{(k)})} \Bigg)^n.
\end{align}

\textit{Remark}: The NZSR and the SOP can be evaluated directly from $F_{\Gamma_{\text{R}}^{(n^*)}}(y)$ in (\ref{gamma_r_ots_with}) following (\ref{eq_NZSR_from_SOP}) and (\ref{eq_new_SOP_from_GammaR}), respectively.

We note from  (\ref{gamma_r_min_without}),  (\ref{gamma_r_min_with}), (\ref{gamma_r_tts_without}), (\ref{gamma_r_tts_with}), (\ref{NZ_OS_without}), and (\ref{gamma_r_ots_with})  that the number of summations and multiplication terms significantly increase with the increase in $N$ and $K$, making the SOP and NZSR evaluation very computationally intensive.
To reduce the computational complexity and to find better insights, we provide the asymptotic analysis of the SOP for the BKU and BKA cases in Section \ref{sec_asymptotic_analysis_sop_unreliable} and the diversity order analysis for perfect backhaul links in Section \ref{section_diversity_order}.

\section{Ergodic Secrecy Rate}
\label{section_ergodic_secrecy_rate}
In this section, the ESR of each transmitter selection scheme is derived in both BKU and BKA cases. The ESR is evaluated using (\ref{ergodic_secrecy_rate}) with the help of the distribution of $\Gamma_{\text{R}}^{(n^*)}$ already derived in Section   \ref{sec_distribution_snr_ratio} for each selection scheme and backhaul link activity knowledge.

\subsection{Minimal eavesdropping selection in the BKU case (MIN-ES-BKU)}
In this section, the transmitter selection is performed according to (\ref{MIN-ES_k}) as already mentioned in Section \ref{sub_sec_min_es_bku} for the MIN-ES-BKU scheme. The ESR is consequently obtained from (\ref{ergodic_secrecy_rate}) with the help of $F_{\Gamma_{\text{R}}^{(n^*)}}(x)$ derived for the same scheme in (\ref{gamma_r_min_without}) as
\begin{align}\label{erg_min_es_without}
&\mathcal{C}_{\text{erg}}=\frac{s}{\ln(2)}\int_{1}^{\infty} \sum_{\mathbf{i}\in\mathcal{M}^{(N)}}\binom{N}{i_1,\ldots,i_K}\nn\\
&\times\Bigg(\prod_{k=1}^{K}\Bigg(\frac{\prod_{m=1}^{K}\lambda_{\text{E}}^{(m)}}{\lambda_{\text{E}}^{(k)}\prod\limits_{\substack{j=1\\j\ne k}}^K(\lambda_{\text{E}}^{(j)}-\lambda_{\text{E}}^{(k)})}\Bigg)^{i_k}\Bigg)\frac{\tilde{\lambda}_{\text{E}}^{(k)}e^{-\lambda_{\text{D}} (x-1)}}{x(x\lambda_{\text{D}} +\tilde{\lambda}_{\text{E}}^{(k)})}dx.
\end{align}
The integral in (\ref{erg_min_es_without}) is evaluated by first adopting partial fraction and then directly using  in \cite[eq.
(3.352.2)]{table_of_integrals} as
\begin{align}\label{erg_min_es_without_final}
&\mathcal{C}_{\text{erg}}=
% \frac{s}{\ln(2)}\int_{1}^{\infty} \sum_{i_1+\ldots+i_K=N}\binom{N}{i_1,\ldots,i_K}\Bigg(\prod_{k=1}^{K}\Bigg(\frac{\prod_{m=1}^{K}\lambda_{\text{E}}^{(m)}}{\lambda_{\text{E}}^{(k)}\prod\limits_{\substack{j=1\\j\ne k}}^K(\lambda_{\text{E}}^{(j)}-\lambda_{\text{E}}^{(k)})}\Bigg)^{i_k}\Bigg)\frac{\tilde{\lambda}_{\text{E}}^{(k)}e^{-\lambda_{\text{D}} (x-1)}}{x(x\lambda_{\text{D}} +\tilde{\lambda}_{\text{E}}^{(k)})}dx\nn\\ 
%&=\frac{1}{\ln(2)}\int_{1}^{\infty}\sum_{i_1+\ldots+i_K=N}\binom{N}{i_1,\ldots,i_K}\Bigg(\prod_{k=1}^{K}\Bigg(\frac{\prod_{m=1}^{K}\lambda_{\text{E}}^{(m)}}{\lambda_{\text{E}}^{(k)}\prod\limits_{\substack{j=1\\j\ne k}}^K(\lambda_{\text{E}}^{(j)}-\lambda_{\text{E}}^{(k)})}\Bigg)^{i_k}\Bigg)\frac{s\tilde{\lambda}_{\text{E}}^{(k)}e^{-\lambda_{\text{D}} (x-1)}}{z(x\lambda_{\text{D}} +\tilde{\lambda}_{\text{E}}^{(k)})}dx\nn\\
% &=\frac{1}{\ln(2)}\sum_{i_1+\ldots+i_K=N}\binom{N}{i_1,\ldots,i_K}\Bigg(\prod_{k=1}^{K}\Bigg(\frac{\prod_{m=1}^{K}\lambda_{\text{E}}^{(m)}}{\lambda_{\text{E}}^{(k)}\prod\limits_{\substack{j=1\\j\ne k}}^K(\lambda_{\text{E}}^{(j)}-\lambda_{\text{E}}^{(k)})}\Bigg)^{i_k}\Bigg)e^{\lambda_{\text{D}}}sa_t \int_{1}^{\infty}\frac{e^{-\lambda_{\text{D}} x}}{z(x+a_t)}dx\nn\\
\frac{s}{\ln(2)}\sum_{\mathbf{i}\in\mathcal{M}^{(N)}}\binom{N}{i_1,\ldots,i_K}\nn\\
&\times\Bigg(\prod_{k=1}^{K}\Bigg(\frac{\prod_{m=1}^{K}\lambda_{\text{E}}^{(m)}}{\lambda_{\text{E}}^{(k)}\prod\limits_{\substack{j=1\\j\ne k}}^K(\lambda_{\text{E}}^{(j)}-\lambda_{\text{E}}^{(k)})}\Bigg)^{i_k}\Bigg)\frac{e^{\lambda_{\text{D}}}\tilde{\lambda}_{\text{E}}^{(k)}}{\lambda_{\text{D}}}\nn\\
&\times \Bigg(\int_{1}^{\infty}\frac{\lambda_{\text{D}}e^{-\lambda_{\text{D}} x}}{\tilde{\lambda}_{\text{E}}^{(k)}x}dx-\int_{1}^{\infty}\frac{\lambda_{\text{D}}e^{-\lambda_{\text{D}} x}}{\tilde{\lambda}_{\text{E}}^{(k)}\Big(x+\frac{\tilde{\lambda}_{\text{E}}^{(k)}}{\lambda_{\text{D}}}\Big)}dx\Bigg)\nn\\
% &=\frac{1}{\ln(2)}\sum_{i_1+\ldots+i_K=N}\binom{N}{i_1,\ldots,i_K}\Bigg(\prod_{k=1}^{K}\Bigg(\frac{\prod_{m=1}^{K}\lambda_{\text{E}}^{(m)}}{\lambda_{\text{E}}^{(k)}\prod\limits_{\substack{j=1\\j\ne k}}^K(\lambda_{\text{E}}^{(j)}-\lambda_{\text{E}}^{(k)})}\Bigg)^{i_k}\Bigg)e^{\lambda_{\text{D}}}s\nn\\
% &\times \Big(-\Ei(-\lambda_{\text{D}})+e^{a_t\lambda_{\text{D}}}\Ei\Big(-(a_t\lambda_{\text{D}}+\lambda_{\text{D}})\Big)\Big)\nn\\
&=\frac{s}{\ln(2)}\sum_{\mathbf{i}\in\mathcal{M}^{(N)}}\binom{N}{i_1,\ldots,i_K}\nn\\
&\times\Bigg(\prod_{k=1}^{K}\Bigg(\frac{\prod_{m=1}^{K}\lambda_{\text{E}}^{(m)}}{\lambda_{\text{E}}^{(k)}\prod\limits_{\substack{j=1\\j\ne k}}^K(\lambda_{\text{E}}^{(j)}-\lambda_{\text{E}}^{(k)})}\Bigg)^{i_k}\Bigg)\Big(e^{\tilde{\lambda}_{\text{E}}^{(k)}+\lambda_{\text{D}}}\nn\\
&\times\Ei(-(\tilde{\lambda}_{\text{E}}^{(k)}+\lambda_{\text{D}}))-e^{\lambda_{\text{D}}}\Ei(-\lambda_{\text{D}})\Big),
\end{align}
where 
% $a_t=\lambda_{\text{E}}^{(k)}/\lambda_{\text{D}}$ and
$\Ei(x)=\int_{-\infty}^{x}\frac{e^{t}}{t}dt$ is the exponential integral.

%\color{red}

%For this selection scheme, the CDF of ${\Gamma_{\text{R}}^{(n^*)}}$ is expressed as
%\begin{align}\label{gamma_MIN-ES}
%F_{\Gamma_{\text{R}}^{(n^*)}}(x)&=\mathbb{P}\Big[\frac{1+{\Gamma}_{\text{SD}}^{(n)}}{1+\min\limits_{n\in \{1,2\ldots, N\}}\{\hat{\Gamma}_{\text{SE}}^{(n)}\}} \le x\Big].
%\end{align}
% Here we realize that the structure of (\ref{gamma_MIN-ES}) is the same as the already computed SOP expression for MIN-ES-BKU in   (\ref{sop_MIN_ES_without}) where $\rho$ is replaced by a dummy variable $x$.
%Thus using (\ref{sop_MIN-ES_without}), the CDF of ${\Gamma_{\text{R}}^{(n^*)}}$ for this case is expressed as
%\begin{align}\label{gamma_MIN-ES_without}
%&F_{\Gamma_{\text{R}}^{(n^*)}}(x)=1-s\sum_{i_1+\ldots+i_K=N}\binom{N}{i_1,\ldots,i_K}\Bigg(\prod_{k=1}^{K}\Bigg(\frac{\prod_{m=1}^{K}\lambda_{\text{E}}^{(m)}}{\lambda_{\text{E}}^{(k)}\prod\limits_{\substack{j=1\\j\ne k}}^K(\lambda_{\text{E}}^{(j)}-\lambda_{\text{E}}^{(k)})}\Bigg)^{i_k}\Bigg)\frac{\tilde{\lambda}_{\text{E}}^{(k)}e^{-\lambda_{\text{D}} (x-1)}}{(x\lambda_{\text{D}} +\tilde{\lambda}_{\text{E}}^{(k)})} .
%\end{align}

\subsection{Minimal eavesdropping selection in the BKA case (MIN-ES-BKA)}
In this section, the transmitter selection is performed according to (\ref{eq_gain_min_with}) as already mentioned in Section \ref{sub_sec_min_es_bka} for the MIN-ES-BKA scheme.  Following the same procedure as in (\ref{erg_min_es_without_final}), the ESR is evaluated by utilizing  $F_{\Gamma_{\text{R}}^{(n^*)}}(x)$ from (\ref{gamma_r_min_with}) in (\ref{ergodic_secrecy_rate}) as
\begin{align}\label{erg_min_es_with}
&\mathcal{C}_{\text{erg}}=\frac{1}{\ln(2)}\sum_{n=1}^{N}\sum_{\mathbf{i}\in\mathcal{M}^{(n)}}\binom{N}{n}\binom{n}{i_1,\ldots,i_K}\nn\\
&\times\Bigg(\prod_{k=1}^{K}\Bigg(\frac{\prod_{m=1}^{K}\lambda_{\text{E}}^{(m)}}{\lambda_{\text{E}}^{(k)}\prod\limits_{\substack{j=1\\j\ne k}}^K(\lambda_{\text{E}}^{(j)}-\lambda_{\text{E}}^{(k)})}\Bigg)^{i_k}\Bigg)(1-s)^{N-n}s^n\nn\\
&\times\Big(e^{\tilde{\lambda}_{\text{E}}^{(k)}+\lambda_{\text{D}}}\Ei(-(\tilde{\lambda}_{\text{E}}^{(k)}+\lambda_{\text{D}}))-e^{\lambda_{\text{D}}}\Ei(-\lambda_{\text{D}})\Big).
\end{align}

% \color{red}
% and thus the CDF of ${\Gamma_{\text{R}}^{(n^*)}}$ is expressed as
% \begin{align}\label{gamma_MIN-ES}
% F_{\Gamma_{\text{R}}^{(n^*)}}(x)&=\mathbb{P}\Big(\frac{1+{\Gamma}_{\text{SD}}^{(n)}}{1+\min\limits_{n\in\mathcal{S}}\{\hat{\Gamma}_{\text{SE}}^{(n)}\}} \le x\Big).
% \end{align}
% Following (\ref{sop_MIN-ES_with}) from MIN-ES-BKA case under SOP in Section-\ref{sub_sec_sop_mines_bku}, the distribution of ${\Gamma_{\text{R}}^{(n^*)}}$ in this case is expressed as 
% \begin{align}\label{gamma_MIN-ES_with}
% &F_{\Gamma_{\text{R}}^{(n^*)}}(x)=1-\sum_{n=1}^{N}\binom{N}{n}(1-s)^{N-n}s^n\sum_{i_1+\ldots+i_K=n}\binom{n}{i_1,\ldots,i_K}\Bigg(\prod_{k=1}^{K}\Bigg(\frac{\prod_{m=1}^{K}\lambda_{\text{E}}^{(m)}}{\lambda_{\text{E}}^{(k)}\prod\limits_{\substack{j=1\\j\ne k}}^K(\lambda_{\text{E}}^{(j)}-\lambda_{\text{E}}^{(k)})}\Bigg)^{i_k}\Bigg)\nn\\
% &\times\frac{\tilde{\lambda}_{\text{E}}^{(k)}e^{-\lambda_{\text{D}}(x-1)}}{x\lambda_{\text{D}}+\tilde{\lambda}_{\text{E}}^{(k)}}.
% \end{align}
% \color{black}

\subsection{Traditional transmitter selection in the BKU case (TTS-BKU) }
In this section, the transmitter selection is performed according to (\ref{eq_gain_max_without}) as already mentioned in Section \ref{sub_sec_tts_bku} for the TTS-BKU case. The ESR in this case is derived by substituting for $F_{\Gamma_{\text{R}}^{(n^*)}}(x)$ from (\ref{gamma_r_tts_without}) in (\ref{ergodic_secrecy_rate}) as  
\begin{align}\label{ERG_TTS_without}
\mathcal{C}_{\text{erg}}&=\frac{1}{\ln(2)}\int_{1}^{\infty}\frac{1}{x}\Bigg(\sum_{n=1}^{N}\sum_{k=1}^{K}\binom{N}{n}(-1)^{n+1}s\nn\\
&\times\frac{\Big(\prod_{m=1}^{K}\lambda_{\text{E}}^{(m)}\Big)e^{-\lambda_{\text{D}} (x-1)}}{(\lambda_{\text{D}} n x+\lambda_{\text{E}}^{(k)})\prod\limits_{\substack{j=1\\j\ne k}}^K(\lambda_{\text{E}}^{(j)}-\lambda_{\text{E}}^{(k)})}\Bigg)dx.
\end{align}
The integral in (\ref{ERG_TTS_without}) is evaluated by using the partial fraction method and following the steps taken for achieving (\ref{erg_min_es_without_final}) as
\begin{align}\label{eq_ERG_TTS_without}
&\mathcal{C}_{\text{erg}}=\frac{s}{\ln(2)}\sum_{n=1}^{N}\sum_{k=1}^{K}\binom{N}{n}\frac{(-1)^{n+1}\Big(\prod_{m=1}^{K}\lambda_{\text{E}}^{(m)}\Big)}{\lambda_{\text{E}}^{(k)}\prod\limits_{\substack{j=1\\j\ne k}}^K(\lambda_{\text{E}}^{(j)}-\lambda_{\text{E}}^{(k)})}\nn\\
&\times\Big(-e^{n\lambda_{\text{D}}}\Ei (-n\lambda_{\text{D}} )+e^{(\lambda_{\text{E}}^{(k)}+n\lambda_{\text{D}})}\Ei(-(\lambda_{\text{E}}^{(k)}+n\lambda_{\text{D}}))\Big).
\end{align}

\subsection{Traditional transmitter selection in the BKA case (TTS-BKA) }
In this section, the transmitter is selected according to (\ref{eq_gain_max_with}) as already mentioned in Section \ref{sub_sec_tts_bku} for the TTS-BKA case. Thus, the ESR for this case is evaluated by utilizing $F_{{\Gamma}_{\text{R}}^{(n^*)}}(x)$ from (\ref{gamma_r_tts_with}) in (\ref{ergodic_secrecy_rate}) as
\begin{align}\label{ERG_TTS_with}
&\mathcal{C}_{\text{erg}}=\int_{1}^{\infty}\sum_{n=1}^{N}\sum_{k=1}^{K}\sum_{q=1}^{n}\binom{N}{n}\binom{n}{q}(1-s)^{N-n}s^{n}(-1)^{q+1}\nn\\
&\times\frac{\Big(\prod_{m=1}^{K}\lambda_{\text{E}}^{(m)}\Big)e^{-\lambda_{\text{D}} (x-1)}}{x(qx\lambda_{\text{D}} +\lambda_{E})\prod\limits_{\substack{j=1\\j\ne k}}^K(\lambda_{\text{E}}^{(j)}-\lambda_{\text{E}}^{(k)})}dx\nn\\
&=\frac{1}{\ln(2)}\sum_{n=1}^{N}\sum_{k=1}^{K}\sum_{q=1}^{n}\binom{N}{n}\binom{n}{q}(1-s)^{N-n}s^{n}(-1)^{q+1}\nn\\
&\times\frac{\prod_{m=1}^{K}\lambda_{\text{E}}^{(m)}}{\lambda_{\text{E}}^{(k)}\prod\limits_{\substack{j=1\\j\ne k}}^K(\lambda_{\text{E}}^{(j)}-\lambda_{\text{E}}^{(k)})}\Big(-e^{q\lambda_{\text{D}}}\Ei (-q\lambda_{\text{D}})\nn\\
&+e^{\lb(\lambda_{\text{E}}^{(k)}+q\lambda_{\text{D}}\rb)}\Ei(-\lambda_{\text{E}}^{(k)}-q\lambda_{\text{D}})\Big).
\end{align}
\subsection{Optimal transmitter selection in the BKU case (OTS-BKU)}
In this section, the best transmitter is selected according to (\ref{_max_OTS_without}) as shown in Section \ref{sub_sec_ots_bku} for the OTS-BKU case. By substituting $F_{{\Gamma}_{\text{R}}^{(n^*)}}(x)$ from (\ref{gamma_r_ots_without}) in (\ref{ergodic_secrecy_rate}), the ESR is evaluated as
\begin{align}\label{ERG_OTS_without}
&\mathcal{C}_{\text{erg}}=\frac{1}{\ln(2)}\int_{1}^{\infty}\frac{1}{x}\sum_{n=1}^{N}\sum_{\mathbf{i}\in\mathcal{M}^{(n)}}\binom{N}{n}\binom{n}{i_1,\ldots,i_K}(-1)^{n+1}\nn\\
&\times s\Bigg(\prod_{k=1}^{K}\Bigg(\frac{\prod_{m=1}^{K}\lambda_{\text{E}}^{(m)}}{\lambda_{\text{D}}\prod\limits_{\substack{j=1\\j\ne k}}^K(\lambda_{\text{E}}^{(j)}-\lambda_{\text{E}}^{(k)})}\Bigg)^{i_k}\Bigg)\frac{e^{-\lambda_{\text{D}} (x-1)\tilde{i}_k}dx}{\prod_{k=1}^{K}\Big(x+\frac{\lambda_{\text{E}}^{(k)}}{\lambda_{\text{D}}}\Big)^{i_k}},
\end{align}
where $\tilde{i}_k=\sum_{k=1}^{K}i_k$. The integral is evaluated by using  the partial fraction method and then following the similar steps as were taken for (\ref{erg_min_es_without_final}) as
\begin{align}
&\mathcal{C}_{\text{erg}}
% &=\frac{1}{\ln(2)}\sum_{n=1}^{N}\binom{N}{n}(-1)^{n+1}s\sum_{i_1+\ldots+i_K=n}\binom{n}{i_1,\ldots,i_K}\Bigg(\prod_{k=1}^{K}\Bigg(\frac{\prod_{m=1}^{K}\lambda_{\text{E}}^{(m)}}{\lambda_{\text{D}}\prod\limits_{\substack{j=1\\j\ne k}}^K(\lambda_{\text{E}}^{(j)}-\lambda_{\text{E}}^{(k)})}\Bigg)^{i_k}\Bigg)\nn\\
% &\times\Big(\int_{1}^{\infty}\frac{e^{-\lambda_{\text{D}} (x-1)\tilde{i}_k}}{x\prod_{k=1}^{K}a_t^{i_k}}dx+\int_{1}^{\infty}\sum_{k=1}^{K}\sum_{l_k=1}^{i_k}\frac{A_k^{(i_k)} e^{-\lambda_{\text{D}} (x-1)\tilde{i}_k}}{(x+a_t)^{i_k-l_t+1}}dx\Big)\nn\\
=\frac{1}{\ln(2)}\sum_{n=1}^{N}\sum_{\mathbf{i}\in\mathcal{M}^{(n)}}\binom{N}{n}\binom{n}{i_1,\ldots,i_K}(-1)^{n+1}s\nn\\
&\times\Bigg(\prod_{k=1}^{K}\Bigg(\frac{\prod_{m=1}^{K}\lambda_{\text{E}}^{(m)}}{\lambda_{\text{D}}\prod\limits_{\substack{j=1\\j\ne k}}^K(\lambda_{\text{E}}^{(j)}-\lambda_{\text{E}}^{(k)})}\Bigg)^{i_k}\Bigg)\Bigg(-\frac{e^{\lambda_{\text{D}} \tilde{i}_k}\Ei{(-\lambda_{\text{D}} \tilde{i}_k)}}{\prod_{k=1}^{K}\Big(\frac{\lambda_{\text{E}}^{(k)}}{\lambda_{\text{D}}}\Big)^{i_k}}\nn\\
&-\sum_{k=1}^{K}\sum_{l_k=1}^{i_k}{A_k^{(i_k)}(\lambda_{\text{D}} \tilde{i}_k)^{l_k-i_k}e^{((\lambda_{\text{D}}+\lambda_{\text{E}}^{(k)}) \tilde{i}_k)} }\nn\\
&\times\Gamma\Big[l_k-i_k,(-(\lambda_{\text{D}}+\lambda_{\text{E}}^{(k)}) \tilde{i}_k)\Big]\Bigg),
\end{align}
where $\Gamma[m,x]=\int_{x}^{\infty}t^{m-1}e^{-t}dt$ is the incomplete gamma function with $m$ being a positive integer and $A_k^{(i_k)}$ is the partial fraction coefficient easily obtained using the standard partial fraction method.

\subsection{Optimal transmitter selection in the BKA case (OTS-BKA)}
In this section, the best transmitter is selected according to (\ref{max_OTS_with}) as shown in Section \ref{sub_sec_ots_bka} for the OTS-BKA case. The ESR is derived by following the same  procedure as in  (\ref{ERG_OTS_without})  by utilizing (\ref{gamma_r_ots_with}) in (\ref{ergodic_secrecy_rate}) as
\begin{align}\label{ERG_OTS_with}
&\mathcal{C}_{\text{erg}}=\frac{1}{\ln(2)}\sum_{n=1}^{N}\sum_{\mathbf{i}\in\mathcal{M}^{(n)}}\binom{N}{n}\binom{n}{i_1,\ldots,i_K}(-1)^{n+1}s^n\nn\\
&\times\Bigg(\prod_{k=1}^{K}\Bigg(\frac{\prod_{m=1}^{K}\lambda_{\text{E}}^{(m)}}{\lambda_{\text{D}}\prod\limits_{\substack{j=1\\j\ne k}}^K(\lambda_{\text{E}}^{(j)}-\lambda_{\text{E}}^{(k)})}\Bigg)^{i_k}\Bigg)\Bigg(-\frac{e^{\lambda_{\text{D}} \tilde{i}_k}\Ei{(-\lambda_{\text{D}} \tilde{i}_k)}}{\prod_{k=1}^{K}\Big(\frac{\lambda_{\text{E}}^{(k)}}{\lambda_{\text{D}}}\Big)^{i_k}}\nn\\
&-\sum_{k=1}^{K}\sum_{l_k=1}^{i_k}{A_k^{(i_k)}(\lambda_{\text{D}} \tilde{i}_k)^{l_k-i_k}e^{((\lambda_{\text{D}}+\lambda_{\text{E}}^{(k)}) \tilde{i}_k)} }\nn\\
&\times\Gamma\Big[l_k-i_k,(-(\lambda_{\text{D}}+\lambda_{\text{E}}^{(k)}) \tilde{i}_k)\Big]\Bigg).
\end{align}

It is difficult to understand from the exact ESR expressions in (\ref{erg_min_es_without_final}), (\ref{erg_min_es_with}), (\ref{ERG_TTS_without}), and (\ref{ERG_TTS_with}) that how it depends on  $s$, $1/\lambda_{\text{D}}$, $1/\lambda_{\text{E}}^{(k)}$ for all $k\in\{1,\ldots,K\}$, $K$, and $N$. 
To get insights into how the backhaul link reliability factor and other system parameters affect the NZSR and SOP, we provide a simplified asymptotic ESR expression in Section \ref{section_asymptotic_ESR}.

\textit{Remark}: The methodology followed in Section \ref{sec_distribution_snr_ratio} and  Section \ref{section_ergodic_secrecy_rate} is uniform and can be used to find any secrecy performance (SOP, NZSR, and  ESR) for different transmitter selection schemes while incorporating the backhaul uncertainty.

\section{Asymptotic Analysis of SOP}\label{sec_asymptotic_analysis_sop_unreliable}
In this section, the asymptotic behavior of the system is analyzed when the SNR of the $\text{S}_n$-$\text{D}$ link is increased asymptotically compared to the SNR of the eavesdroppers' links. By assuming $1/\lambda_{\text{D}}  \rightarrow \infty$ for a given $1/\lambda_{\text{E}}^{(k)}$ for each $n\in\{1,\dots, K\}$ and $k \in \{1,\ldots, N\}$, the asymptotic analysis of the SOP for the selection schemes in BKU and BKA cases are carried out. For the asymptotic analysis, we first approximate  $F_{\Gamma_{\text{R}}^{(n^*)}}(y)$ already derived in Section \ref{sec_distribution_snr_ratio} for each selection scheme under the condition of unreliable backhaul links when $1/\lambda_{\text{D}}  \rightarrow \infty$ and use the same in the SOP expression of (\ref{eq_new_SOP_from_GammaR}). The asymptotic analysis provides insight into the impact of unreliable backhaul connections on the system's performance in the high-SNR regime. 

\subsection{Minimal eavesdropping selection (MIN-ES) }
In this section, we evaluate the asymptotic SOP of the MIN-ES scheme for both the BKU and BKA cases. To derive the asymptotic SOP for the MIN-ES-BKU scheme, we first find an approximation of $F_{\Gamma_{\text{R}}^{(n^*)}}(y)$ in (\ref{gamma_r_min_without}) as $\textnormal{exp}(-y\lambda_{\text{D}})$ tends to unity when $1/\lambda_{\text{D}}  \rightarrow \infty$ as 
\begin{align}\label{NZ_MIN-ES_without_asym}
\mathcal{P}_{\text{out}}^\infty&=1-s\sum_{i_1+\ldots+i_K=n}\binom{n}{i_1,\ldots,i_K}\nn\\
&\times\Bigg(\prod_{k=1}^{K}\Bigg(\frac{\prod_{m=1}^{K}\lambda_{\text{E}}^{(m)}}{\lambda_{\text{E}}^{(k)}\prod\limits_{\substack{j=1\\j\ne k}}^K(\lambda_{\text{E}}^{(j)}-\lambda_{\text{E}}^{(k)})}\Bigg)^{i_k}\Bigg)
% \frac{\tilde{\lambda}_{\text{E}}^{(k)}}{\tilde{\lambda}_{\text{E}}^{(k)}}\nn\\
= (1-s).
\end{align}
The multiplier of $s$ in the second term of the above equation becomes unity. Then, using (\ref{NZ_MIN-ES_without_asym}) in the SOP expression of (\ref{eq_new_SOP_from_GammaR}), we obtain the asymptotic SOP expression.

Similarly, for the MIN-ES-BKA scheme, $F_{\Gamma_{\text{R}}^{(n^*)}}(y)$ in (\ref{gamma_r_min_with}) is approximated at high-SNR first, and then the asymptotic SOP is evaluated using the SOP expression from (\ref{eq_new_SOP_from_GammaR}) as
\begin{align}\label{NZ_MIN-ES_with_asym}
\mathcal{P}_{\text{out}}^\infty&=1-\sum_{n=1}^{N}\binom{N}{n}(1-s)^{N-n}s^n=(1-s)^N. 
\end{align}
% \textcolor{blue}{It is observed from (\ref{NZ_MIN-ES_without_asym}) and (\ref{NZ_MIN-ES_with_asym}) that the SOP at high-SNR saturates to a constant value depending on the backhaul link reliability factor  $s$ and does not depend on the channel quality.}

\subsection{Traditional transmitter selection (TTS)}
In this section, the asymptotic SOP of the TTS scheme for both the BKU and BKA cases is evaluated. Applying $1/\lambda_{\text{D}}  \rightarrow \infty$  in (\ref{gamma_r_tts_without}) we approximate $F_{\Gamma_{\text{R}}^{(n^*)}}(y)$ for the TTS-BKU scheme, and using the same in the SOP expression of (\ref{eq_new_SOP_from_GammaR}), the asymptotic SOP is evaluated as $\mathcal{P}_{\text{out}}^\infty= (1-s).$
% \begin{align}\label{NZ_TTS_without_asym}
% \mathcal{P}_{\text{out}}^\infty= (1-s).
% \end{align}

Similarly, for the TTS-BKA scheme, the asymptotic SOP is evaluated using the approximate  $F_{\Gamma_{\text{R}}^{(n^*)}}(y)$ from (\ref{gamma_r_tts_with}) at high-SNR  and then applying it in (\ref{eq_new_SOP_from_GammaR}) as
\begin{align}\label{NZ_TTS_with_asym}
\mathcal{P}_{\text{out}}^\infty&= 1-\sum_{n=1}^{N}\binom{N}{n}(-1)^{n+1}s^n=(1-s)^N.
\end{align}
% \textcolor{blue}{It is also observed in the TTS scheme that the asymptotes in BKU and BKA cases do not depend on the channel quality but  on the backhaul link reliability factor  $s$ at high-SNR.}
\subsection{Optimal transmitter selection (OTS)}
The asymptotic SOP for the OTS-BKU scheme can easily be expressed following the same approach as in (\ref{NZ_MIN-ES_without_asym}) using the approximate $F_{\Gamma_{\text{R}}^{(n^*)}}(y)$ from (\ref{gamma_r_ots_without}) at high-SNR and then using (\ref{eq_new_SOP_from_GammaR}) as 
\begin{align}\label{asym_NZ_OS_without}
&\mathcal{P}_{\text{out}}^\infty= 1- s\sum_{n=1}^{N}\binom{N}{n}(-1)^{n+1}\nn\\
&\times\Bigg(\sum_{k=1}^{K}\frac{\prod_{m=1}^{K}\lambda_{\text{E}}^{(m)}}{\lambda_{\text{E}}^{(k)}\prod\limits_{\substack{j=1\\j\ne k}}^K(\lambda_{\text{E}}^{(j)}-\lambda_{\text{E}}^{(k)})} \Bigg)^n =(1-s).
\end{align}

The asymptotic SOP for the OTS-BKA scheme is derived using the approximated $F_{\Gamma_{\text{R}}^{(n^*)}}(y)$ from (\ref{gamma_r_ots_with}) at high-SNR  and then using (\ref{eq_new_SOP_from_GammaR}) as  $\mathcal{P}_{\text{out}}^\infty=\lb(1-s\rb)^N$.
% \begin{align}\label{asm_NZ_OS_with}
% \mathcal{P}_{\text{out}}^\infty&=\lb(1-s\rb)^N.
% \end{align}

 The asymptotic analysis of the NZSR can be evaluated directly from the asymptotic SOP derived in this section for each selection scheme and backhaul link activity knowledge case following (\ref{eq_NZSR_from_SOP}).

It is observed from the asymptotes derived for each case in this section  that the SOP at high-SNR saturates to a constant value depending on the backhaul link reliability factor  $s$ and does not depend on the channel quality. We note that the asymptotic SOP is the same, i.e., $(1-s)$, in the BKU case, irrespective of the selection schemes. Similarly, the asymptotic SOP in the BKA case for all the selection schemes is the same, i.e., $(1-s)^N$. We further notice that when backhaul link activity knowledge is unavailable, the asymptotic SOP can not be improved further by increasing the number of transmitters. In contrast, the SOP improvement is possible when the backhaul link activity knowledge is available by increasing the number of transmitters. We also observe that asymptotic SOP requires very few computations compared to the exact SOP expression derived in Section \ref{sec_distribution_snr_ratio}.

\section{Diversity Order Analysis of SOP with  Perfect Backlhaul Links}\label{section_diversity_order}

In this section, we evaluate the diversity order for each transmitter selection scheme when all backhaul links are perfect, i.e., $s=1$. 
Since the asymptotic values depend only on the backhaul link reliability factor,  as shown in the previous section, it is not possible to understand how the SOP reaches the asymptotic values, which is generally answered by the secrecy diversity order analysis. That is why we derive the diversity order here when $s=1$.
The diversity order is  defined as the negative slope of the SOP curve at a high-SNR regime defined in Section \ref{sec_asymptotic_analysis_sop_unreliable} in the logarithm scale of the SNR and is expressed as
\begin{align}\label{diversity_order}
d=-\lim_{1/\lambda_{\text{D}}\rightarrow\infty} \frac{\log\lb(\mathcal{P}_{\text{out}}\rb)}{\log\lb(1/\lambda_{\text{D}}\rb)}.    
\end{align}
Toward the goal, we first obtain an approximate SOP corresponding to the selection schemes at the high-SNR regime and then use the diversity order definition in (\ref{diversity_order}). To obtain an approximate $\mathcal{P}_{\text{out}}$  corresponding to the transmitter selection scheme in the high-SNR regime, we derive  approximate $F_{\Gamma_{\text{R}}^{(n^*)}}(\rho)$ following Section \ref{sec_asymptotic_analysis_sop_unreliable} but in the case of perfect backhaul links using the first-order Taylor series approximation 
$F_{\Gamma^{(n)}_{\text{SD}}}(x)=\lambda_\text{D}x$ from (\ref{cdf_SD}) and $f_ {{\Gamma}_{\text{SE}}^{(n^*)}}(x)$ from Section  \ref{sec_asymptotic_analysis_sop_unreliable} for perfect backhaul links. Under the perfect backhaul links condition, the notion of BKU and BKA does not apply.

\subsection{Minimal eavesdropping selection (MIN-ES)}
To obtain the diversity order, as described above, $F_{\Gamma_{\text{SD}}^{(n)}}$ in (\ref{cdf_SD}) is approximated first using first-order Taylor series approximation as $F_{\Gamma^{(n)}_{\text{SD}}}(x)=\lambda_\text{D}x$ at the high-SNR regime. Next, the approximate SOP is evaluated using $f_{\Gamma_{\text{SE}}^{(n^*)}}(x)$ from  (\ref{PDF_SE_MIN-ES}) under the perfect backhaul links condition and approximate  $F_{\Gamma^{(n)}_{\text{SD}}}(x)$. Following (\ref{gamma_r_min_without_0}), the approximate SOP is
\begin{align}\label{sop_mines_high_snr}
&\mathcal{P}_{\text{out}}=F_{\Gamma_{\text{R}}^{(n^*)}}(\rho)
= \int_{0}^{\infty} F_{{\Gamma}^{(n)}_{\text{SD}}}(\rho(x+1)-1) f_ {{\Gamma}_{\text{SE}}^{(n^*)}}(x) dx\nn\\
% &=\int_{0}^{\infty}\lambda_{\text{D}}(\rho(x+1)-1)\sum_{i_1+\ldots+i_K=N}\binom{N}{i_1,\ldots,i_K}\Bigg(\prod_{k=1}^{K}\Bigg(\frac{\prod_{m=1}^{K}\lambda_{\text{E}}^{(m)}}{\lambda_{\text{E}}^{(k)}\prod\limits_{\substack{j=1\\j\ne k}}^K(\lambda_{\text{E}}^{(j)}-\lambda_{\text{E}}^{(k)})}\Bigg)^{i_k}\Bigg)\tilde{\lambda}_{\text{E}}^{(k)} e^{-\tilde{\lambda}_{\text{E}}^{(k)}x}dx \nn\\
&=\lambda_{\text{D}}\sum_{i_1+\ldots+i_K=N}\binom{N}{i_1,\ldots,i_K}\nn\\
&\times\Bigg(\prod_{k=1}^{K}\Bigg(\frac{\prod_{m=1}^{K}\lambda_{\text{E}}^{(m)}}{\lambda_{\text{E}}^{(k)}\prod\limits_{\substack{j=1\\j\ne k}}^K(\lambda_{\text{E}}^{(j)}-\lambda_{\text{E}}^{(k)})}\Bigg)^{i_k}\Bigg)\Bigg(\frac{1}{\tilde{\lambda}_{\text{E}}^{(k)}}+(\rho-1)\Bigg).
\end{align}
% \textcolor{red}{Why 1? logical explanation??}
% \textcolor{blue}
{Thereafter, using (\ref{sop_mines_high_snr}) in (\ref{diversity_order}), we derive the diversity order as $d=1$.  The diversity order is one because the transmitter selection is carried out based on the eavesdroppers' links. The minimum SNR at the eavesdropper does not guarantee better destination channel, that is why the diversity does not improve. }

% \begin{align}
% d&=-\frac{\log\Big(\lambda_{\text{D}}\Big)}{\log\Big(1/\lambda_{\text{D}}\Big)}=1.    
% \end{align}
\subsection{Traditional transmitter selection (TTS)}
In this case, $F_{\Gamma_{\text{SD}}^{(n*)}}(x)$ in (\ref{CDF_TTS_without_max}) for the TTS scheme is approximated using first-order Taylor series approximation as $F_{\Gamma^{(n*)}_{\text{SD}}}(x)=(\lambda_\text{D}x)^N$.  Next, to evaluate the approximate SOP, we use  $f_{\Gamma_{\text{SE}}^{(n)}}(x)$ from (\ref{PDF_SE}) and approximate $F_{\Gamma^{(n*)}_{\text{SD}}}(x)$ under the perfect backhaul links condition. We obtain the approximate SOP following (\ref{gamma_r_tts_without}) as 
\begin{align}\label{sop_tts_high_snr}
\mathcal{P}_{\text{out}}
% = \int_{0}^{\infty} F_{{\Gamma}^{(n^*)}_{\text{SD}}}(\rho(x+1)-1) f_ {\hat{\Gamma}_{\text{SE}}^{(n)}}(x) dx\nn\\
% &=\int_{0}^{\infty}\lambda_{\text{D}}^N(\rho(x+1)-1)^N\sum_{k=1}^{K}\frac{\prod_{m=1}^{K}\lambda_{\text{E}}^{(m)}}{\prod\limits_{\substack{j=1\\j\ne k}}^K(\lambda_{\text{E}}^{(j)}-\lambda_{\text{E}}^{(k)})} e^{-\lambda_{\text{E}}^{(k)} x}dx \nn\\
&=\lambda_{\text{D}}^N\sum_{n=0}^{N}\sum_{k=1}^{K}\frac{\rho^n(\rho-1)^{N-n}\Big(\prod_{m=1}^{K}\lambda_{\text{E}}^{(m)}\Big)}{\prod\limits_{\substack{j=1\\j\ne k}}^K(\lambda_{\text{E}}^{(j)}-\lambda_{\text{E}}^{(k)})} \Bigg(\frac{n!}{{\lambda}_{\text{E}_k}^{n+1}}\Bigg).
\end{align}
% \textcolor{red}{Why N? logical explanation??}
% \textcolor{blue}
{Then, the diversity order is obtained using (\ref{sop_tts_high_snr}) in (\ref{diversity_order}) as $d=N$, which is intuitive as the transmitter is selected based on the selection among $N$ number of links to the destination.}

% \begin{align}
% d
% % &=-\lim_{1/\lambda_{\text{D}}\rightarrow\infty} \frac{\log\Big(\lambda_{\text{D}}^N\sum_{n=0}^{N}(\rho-1)^{N-n}\rho^n\sum_{k=1}^{K}\frac{\prod_{m=1}^{K}\lambda_{\text{E}}^{(m)}}{\prod\limits_{\substack{j=1\\j\ne k}}^K(\lambda_{\text{E}}^{(j)}-\lambda_{\text{E}}^{(k)})} \Big(\frac{n!}{{\lambda}_{\text{E}_k}^{n+1}}\Big)\Big)}{\log\Big(1/\lambda_{\text{D}}\Big)}\nn\\
% &=-\frac{N\log\Big(\lambda_{\text{D}}\Big)}{\log\Big(1/\lambda_{\text{D}}\Big)}=N.    
% \end{align}
\subsection{Optimal transmitter selection  (OTS)}
Similar to the previous section, $F_{\Gamma_{\text{SD}}^{(n)}}(x)$ in (\ref{cdf_SD}) is approximated first using first-order Taylor series approximation $F_{\Gamma^{(n)}_{\text{SD}}}(x)=\lambda_\text{D}x$ and then the approximate SOP is obtained using approximate  $F_{\Gamma^{(n)}_{\text{SD}}}(x)$ and $f_{\Gamma_{\text{SE}}^{(n)}}$ from (\ref{PDF_SE}) under the perfect backhaul links condition. The approximate SOP following (\ref{gamma_r_ots_without}) is 
\begin{align}\label{sop_ots_high_snr}
\mathcal{P}_{\text{out}}
% = \Big(\int_{0}^{\infty} F_{{\Gamma}^{(n^*)}_{\text{SD}}}(\rho(x+1)-1) f_ {\hat{\Gamma}_{\text{SE}}^{(n)}}(x) dx\Big)^N\nn\\
&=\Bigg(\sum_{k=1}^{K}\frac{\lambda_{\text{D}}\Big(\prod_{m=1}^{K}\lambda_{\text{E}}^{(m)}\Big)\Big(\frac{1}{{\lambda}_{E}^{(k)}}+(\rho-1)\Big)}{\prod\limits_{\substack{j=1\\j\ne k}}^K(\lambda_{\text{E}}^{(j)}-\lambda_{\text{E}}^{(k)})} \Bigg)^N.
\end{align}
% \textcolor{red}{Why N? logical explanation??}
% \textcolor{blue}
{Finally, the application of (\ref{sop_ots_high_snr}) in (\ref{diversity_order}) provides the diversity order as $d=N$, which is the same as the diversity order of the TTS scheme. It can be observed that at high SNRs, both TTS and OTS will relish the same rate of improvement with $N$. However, in OTS, the global CSI is required, whereas the CSI of the eavesdroppers' links is not required in the TTS scheme.}

% \begin{align}
% d&=-\frac{N\log\Big(\lambda_{\text{D}}\Big)}{\log\Big(1/\lambda_{\text{D}}\Big)}=N.    
% \end{align}

\section{Asymptotic Analysis of ESR}\label{section_asymptotic_ESR}
In this section, we provide the asymptotic analysis of the ESR for each selection scheme in both the backhaul link activity knowledge cases, BKU and BKA. The definition of the asymptotic analysis is the same as in Section \ref{sec_asymptotic_analysis_sop_unreliable}, i.e., $1/\lambda_{\text{D}}  \rightarrow \infty$ for a given $1/\lambda_{\text{E}}^{(k)}$ for each $n\in\{1,\dots, K\}$ and $k \in \{1,\ldots, N\}$ unless otherwise specified.

\subsection{Minimal eavesdropping selection (MIN-ES)}
The asymptotic ESR for the MIN-ES-BKU scheme is evaluated first. To derive the asymptotic expression, we approximate  exact $\mathcal{C}_{\text{erg}}$ from (\ref{erg_min_es_without}) derived in Section \ref{section_ergodic_secrecy_rate} for the MIN-ES-BKU scheme when $1/\lambda_{\text{D}}  \rightarrow \infty$. As $1/\lambda_{\text{D}}  \rightarrow \infty$,  the ESR in (\ref{erg_min_es_without}) is  approximated by assuming $\lambda_{\text{D}}  \rightarrow 0$ as
\begin{align}\label{erg_asymp_min_es}
\mathcal{C}_{\text{erg}}&
\approx \frac{s}{\ln(2)}\sum_{\mathbf{i}\in\mathcal{M}^{(N)}}\binom{N}{i_1,\ldots,i_K}\nn\\
&\times\Bigg(\prod_{k=1}^{K}\Bigg(\frac{\prod_{m=1}^{K}\lambda_{\text{E}}^{(m)}}{\lambda_{\text{E}}^{(k)}\prod\limits_{\substack{j=1\\j\ne k}}^K(\lambda_{\text{E}}^{(j)}-\lambda_{\text{E}}^{(k)})}\Bigg)^{i_k}\Bigg)\nn\\
&\times\Big(e^{\tilde{\lambda}_{\text{E}}^{(k)}}\Ei(-\tilde{\lambda}_{\text{E}}^{(k)})-\Ei(-\lambda_{\text{D}})\Big).
\end{align}
Next, we use the definition $\Ei(-x)=C+\ln{(x)}+\int_{0}^{x}\frac{e^{-t}-1}{t}dt$, where 
$C=\lim\limits_{m\rightarrow\infty}\lb(-\log(m)+\sum_{k=1}^{m}\frac{1}{k}\rb)$ = $0.5772$ (with $m$ being a positive integer)  is the Euler constant, from \cite[eq. (8.212.1)]{table_of_integrals} to further approximate  (\ref{erg_asymp_min_es}) by neglecting the integral $\int_{0}^{\lambda_\text{D}}\frac{e^{-t}-1}{t}dt$ while $\lambda_\text{D}\rightarrow 0$. Hence, we write the asymptotic ESR as 
\begin{align}
\label{eq_asy_ESR_MINESMKU}
\mathcal{C}_{\text{erg}}^{\infty}
% &=\frac{s}{\ln(2)}\sum_{\mathbf{i}\in\mathcal{M}^{(N)}}\binom{N}{i_1,\ldots,i_K}\Bigg(\prod_{k=1}^{K}\Bigg(\frac{\prod_{m=1}^{K}\lambda_{\text{E}}^{(m)}}{\lambda_{\text{E}}^{(k)}\prod\limits_{\substack{j=1\\j\ne k}}^K(\lambda_{\text{E}}^{(j)}-\lambda_{\text{E}}^{(k)})}\Bigg)^{i_k}\Bigg)\nn\\
% &\times\Big(e^{\tilde{\lambda}_{\text{E}}^{(k)}}\Ei(-\tilde{\lambda}_{\text{E}}^{(k)})-(C+\ln{(\lambda_{\text{D}})})\Big)\nn\\
&=\frac{s}{\ln(2)}\Bigg(\ln{\Big(\frac{1}{\lambda_{\text{D}}}\Big)}-\sum_{\mathbf{i}\in\mathcal{M}^{(N)}}\binom{N}{i_1,\ldots,i_K}\nn\\
&\times\Bigg(\prod_{k=1}^{K}\Bigg(\frac{\prod_{m=1}^{K}\lambda_{\text{E}}^{(m)}}{\lambda_{\text{E}}^{(k)}\prod\limits_{\substack{j=1\\j\ne k}}^K(\lambda_{\text{E}}^{(j)}-\lambda_{\text{E}}^{(k)})}\Bigg)^{i_k}\Bigg)\Bigg.\nn\\
&\Bigg.\times\Big(C-e^{\tilde{\lambda}_{\text{E}}^{(k)}}\Ei(-\tilde{\lambda}_{\text{E}}^{(k)})\Big)\Bigg).
\end{align}
The asymptotic ESR in (\ref{eq_asy_ESR_MINESMKU}) can be rewritten as a linear function of $\ln{\lb(1/\lambda_{\text{D}}\rb)}$ to derive insights from the equation as
\begin{align}\label{eq_straight_line_form_ESR_MINESBKU}
\mathcal{C}_{erg}^{\infty}&=S^{\infty}\lb(\ln{(1/\lambda_{\text{D}})}-L^{\infty}\rb),
\end{align}
where $S^{\infty}$ is the high-SNR slope  and $L^{\infty}$ is the power offset parameter \cite{wang_PLS_2014}. The slope shows the rate of change of the ESR with SNR at high SNR, and a system should have a high slope. The offset parameter represents the shift of the asymptotic ESR  from the origin, and it is desirable to have $L^\infty$ as small as possible. From (\ref{eq_asy_ESR_MINESMKU}), $S^{\infty}= \frac{s}{\ln(2)} $
% \begin{align}\label{s_infinity_min_without}
% S^{\infty}&=\frac{s}{\ln(2)} 
% \end{align} 
and
\begin{align}\label{l_infinity_min_without}
L^{\infty}
&=C-\sum_{\mathbf{i}\in\mathcal{M}^{(N)}}\binom{N}{i_1,\ldots,i_K}\nn\\
&\times\Bigg(\prod_{k=1}^{K}\Bigg(\frac{\prod_{m=1}^{K}\lambda_{\text{E}}^{(m)}}{\lambda_{\text{E}}^{(k)}\prod\limits_{\substack{j=1\\j\ne k}}^K(\lambda_{\text{E}}^{(j)}-\lambda_{\text{E}}^{(k)})}\Bigg)^{i_k}\Bigg)e^{\tilde{\lambda}_{\text{E}}^{(k)}}\Ei(-\tilde{\lambda}_{\text{E}}^{(k)}).
\end{align}

Following the same procedure as in the MIN-ES-BKU case in obtaining (\ref{eq_straight_line_form_ESR_MINESBKU}) from (\ref{erg_asymp_min_es}), the asymptotic ESR for the MIN-ES-BKA case is evaluated from (\ref{erg_min_es_with}). In this case,  $S^{\infty}=\frac{1-(1-s)^N}{\ln(2)}$
% \begin{align}
% \label{eq_asy_ESR_MINESBKA}
% \mathcal{C}_{\text{erg}}^{\infty}
% &=\frac{1}{\ln(2)}\sum_{n=1}^{N}\binom{N}{n}(1-s)^{N-n}s^n\sum_{i_1+\ldots+i_K=n}\binom{n}{i_1,\ldots,i_K}\prod_{t=1}^{K}\Bigg(\frac{\lambda_{E}^{(1)}\ldots\lambda_{E}^{(n,K)}}{\prod_{\substack{j=1\\j\ne t}}^K\lambda_{E}^{(n,t)}(\lambda_{E}^{(n,j)}-\lambda_{E}^{(n,t)})}\Bigg)^{i_t}\nn\\
% &\times\Big(e^{\tilde{\lambda}_{\text{E}}^{(k)}}\Ei(-\tilde{\lambda}_{\text{E}}^{(k)})-(C+\ln{(\lambda_{\text{D}})})\Big),
% \end{align}
% where 
% \begin{align}\label{s_infinity_min_with}
% S^{\infty}
% % &=\frac{1}{\ln(2)}\sum_{n=1}^{N}\binom{N}{n}(1-s)^{N-n}s^n\sum_{i_1+\ldots+i_K=n}\binom{n}{i_1,\ldots,i_K}\prod_{k=1}^{K}\Bigg(\frac{\prod_{m=1}^{K}\lambda_{\text{E}}^{(m)}}{\prod\limits_{\substack{j=1\\j\ne k}}^K\lambda_{\text{E}}^{(k)}(\lambda_{\text{E}}^{(j)}-\lambda_{\text{E}}^{(k)})}\Bigg)^{i_k}\nn\\
% % &=\frac{1}{\ln(2)}\sum_{n=1}^{N}\binom{N}{n}(1-s)^{N-n}s^n
% =\frac{1-(1-s)^N}{\ln(2)}
% \end{align}
and
\begin{align}\label{l_infinity_min_with}
L^{\infty}&=\frac{1}{1-(1-s)^N}\sum_{n=1}^{N}\sum_{\mathbf{i}\in\mathcal{M}^{(n)}}\binom{N}{n}\binom{n}{i_1,\ldots,i_K}\nn\\
&\times(1-s)^{N-n}s^n\Bigg(\prod_{k=1}^{K}\Bigg(\frac{\prod_{m=1}^{K}\lambda_{\text{E}}^{(m)}}{\prod\limits_{\substack{j=1\\j\ne k}}^K\lambda_{\text{E}}^{(k)}(\lambda_{\text{E}}^{(j)}-\lambda_{\text{E}}^{(k)})}\Bigg)^{i_k}\Bigg)\nn\\
&\times\Big(C-e^{\tilde{\lambda}_{\text{E}}^{(k)}}\Ei(-\tilde{\lambda}_{\text{E}}^{(k)})\Big).
\end{align}
% You have not derived any insight from any of these asymptotic expressions. Please do that.

\subsection{Traditional transmitter selection (TTS)}
In the TTS-BKU scheme, the asymptotic ESR is evaluated by following the same steps as taken to arrive at (\ref{eq_straight_line_form_ESR_MINESBKU}) from (\ref{erg_asymp_min_es}). In this case, we also apply $1/\lambda_{\text{D}}  \rightarrow \infty$ or $\lambda_{\text{D}}  \rightarrow 0$  in (\ref{eq_ERG_TTS_without})  to get the approximated ESR as
\begin{align}\label{eq_asy_ESR_TTSBKU}
\mathcal{C}_{\text{erg}}&\approx \frac{s}{\ln(2)}\sum_{n=1}^{N}\sum_{k=1}^{K}\binom{N}{n}(-1)^{n+1}\frac{\lb(\prod_{m=1}^{K}\lambda_{\text{E}}^{(m)}\rb)}{\lambda_{\text{E}}^{(k)}\prod\limits_{\substack{j=1\\j\ne k}}^K(\lambda_{\text{E}}^{(j)}-\lambda_{\text{E}}^{(k)})}\nn\\
&\times \Big(e^{\lambda_{\text{E}}^{(k)}}\Ei(-\lambda_{\text{E}}^{(k)})-\Ei (-n\lambda_{\text{D}} )\Big).
\end{align}
Using the  same  approximation of $\Ei(-n\lambda_{\text{D}})\approx C+\ln{(n\lambda_{\text{D}})}$ when $\lambda_{\text{D}}\rightarrow 0$ as in the MIN-ES scheme, the asymptotic ESR is evaluated as in (\ref{eq_straight_line_form_ESR_MINESBKU})
% \begin{align}\label{eq_st_lin_form_ESR_TTSBKU}
% \mathcal{C}_{erg}^{\infty}&=S^{\infty}\lb(\ln{(1/\lambda_{\text{D}})}-L^{\infty}\rb),
% \end{align}
where $S^{\infty}= \frac{s}{\ln(2)} $
% \begin{align}\label{s_inf_tts_without}
% S^{\infty}= \frac{s}{\ln(2)} 
% \end{align}
% $S^{\infty}=s$ 
and 
\begin{align}\label{l_inf_tts_without}
L^{\infty}&=C+\sum_{n=1}^{N}\sum_{k=1}^{K}\binom{N}{n}\frac{(-1)^{n+1}\lb(\prod_{m=1}^{K}\lambda_{\text{E}}^{(m)}\rb)}{\lambda_{\text{E}}^{(k)}\prod\limits_{\substack{j=1\\j\ne k}}^K(\lambda_{\text{E}}^{(j)}-\lambda_{\text{E}}^{(k)})}\nn\\
&\times\Big(\ln(n)-e^{\lambda_{\text{E}}^{(k)}}\Ei(-\lambda_{\text{E}}^{(k)})\Big).
% \nn\\
% &=C-e^{\lambda_{\text{E}}^{(k)}}\Ei(-\lambda_{\text{E}}^{(k)})-\sum_{n=1}^{N}\binom{N}{n}(-1)^{n+1}\ln(n).
\end{align}
Similarly, in the TTS-BKA scheme, the  asymptotic ESR is achieved from (\ref{ERG_TTS_with}) 
% with the help of (\ref{ERG_TTS_with})  as
% \begin{align}\label{eq_st_lin_form_ESR_TTSBKA}
% \mathcal{C}_{erg}^{\infty}&=S^{\infty}\lb(\ln{(1/\lambda_{\text{D}})}-L^{\infty}\rb),
% \end{align}
where $\frac{1-\lb(1-s\rb)^N}{\ln(2)} $
% \begin{align}\label{s_inf_tts_with}
% S^{\infty}
% % &=\frac{1}{\ln(2)}\lb(1-\lb(1-s\rb)^N\rb)\sum_{k=1}^{K}\frac{\prod_{m=1}^{K}\lambda_{\text{E}}^{(m)}}{\lambda_{\text{E}}^{(k)}\prod\limits_{\substack{j=1\\j\ne k}}^K(\lambda_{\text{E}}^{(j)}-\lambda_{\text{E}}^{(k)})}\nn\\
% &=\frac{1-\lb(1-s\rb)^N}{\ln(2)},
% \end{align}
and
\begin{align}\label{l_inf_tts_with}
&L^{\infty}=\frac{1}{1-\lb(1-s\rb)^N}\sum_{n=1}^{N}\sum_{k=1}^{K}\sum_{q=1}^{n}\binom{N}{n}\binom{n}{q}(1-s)^{N-n}s^{n}\nn\\
&\times\frac{(-1)^{q+1}\lb(\prod_{m=1}^{K}\lambda_{\text{E}}^{(m)}\rb)}{\lambda_{\text{E}}^{(k)}\prod\limits_{\substack{j=1\\j\ne k}}^K(\lambda_{\text{E}}^{(j)}-\lambda_{\text{E}}^{(k)})}\Big(C+\ln(q)-e^{\lambda_{\text{E}}^{(k)}}\Ei(-\lambda_{\text{E}}^{(k)})\Big).
\end{align}

It is observed in this section that the slope of the ESR for each selection scheme depends only on $s$ and increases as $s$ tends to unity. The highest slope $S^\infty=\frac{1}{\ln(2)}$ is achieved when $s=1$. We also observe that the slopes of the MIN-ES and TTS schemes are the same for the BKU  cases. The slopes in MIN-ES-BKA  and TTS-BKA cases  are also the same; however, higher than the slopes of the MIN-ES-BKU  and TTS-BKU cases. 
This shows that the slope depends on the availability of the backhaul link activity knowledge for a given selection scheme. However, for the given backhaul link activity knowledge case, it is the same irrespective of the selection schemes. 
In the BKU cases, the slope only depends on $s$ but not on $N$. In contrast, the slope in the BKA case depends on $s$ and $N$ both.
This suggests that the slope can be improved by increasing $N$ in the BKA case but not in the BKU case. The slope in the BKU and BKA cases for both the MIN-ES and TTS schemes is independent of $K$.

We note from  (\ref{l_infinity_min_without}), (\ref{l_infinity_min_with}), (\ref{l_inf_tts_without}), and (\ref{l_inf_tts_with}) that $L^\infty$ is independent of $1/\lambda_{\text{D}}$.  It increases as $K$ and $1/\lambda_{\text{E}}^{(k)}$ for each $k\in\{1,\ldots,K\}$ increases. Thus, the asymptotic ESR decreases. This is intuitive as the increase in eavesdropping quality degrades the system's secrecy. We also note  that $L^\infty$  decreases as $N$ increases. From (\ref{l_infinity_min_without}) and (\ref{l_inf_tts_without}), we  observe that in both  the BKU cases for MIN-ES and TTS schemes, $L^\infty$ is independent of $s$ and only depends on $N$, $K$, and $1/\lambda_{\text{E}}^{(k)}$. However, in the BKA cases for MIN-ES and TTS schemes in (\ref{l_infinity_min_with}) and (\ref{l_inf_tts_with}), $L^\infty$ decreases as $s$ tends to unity.
% We also observe from (\ref{l_infinity_min_without}) and (\ref{l_inf_tts_without}) that $L^\infty$ is independent of $s$ in the both BKU and BKA cases for the MIN-ES as well as the TTS scheme.

\subsection{Optimal transmitter selection (OTS)}
In the OTS scheme, it is not trivial to find the asymptotic ESR by using the same approximation as was used for the MIN-ES and TTS schemes. Therefore, the asymptotic ESR for the $n$-th transmitter is derived assuming ${\Gamma}_{\text{SD}}^{(n)}$ and ${\Gamma}_{\text{SE}}^{(n)}$ both operate in the high-SNR regime and $1/\lambda_{\text{D}} >>1/\lambda_{\text{E}}^{(k)}$ for all $n\in\{1,\ldots,N\}$ and $k\in \{1,\ldots,K\}$. In this case, $\Gamma_{\text{R}}^{(n)}$ is approximated by neglecting unity from both the numerator and the denominator in (\ref{gamma_r}). The CDF of ${\Gamma_{\text{R}}^{(n)}}$  with high-SNR approximation when the backhaul link is active, i.e., $s=1$ is written as 
\begin{align}\label{gamma_r_high_snr}
F_{\Gamma_{\text{R}}^{(n)}}(x)&=\mathbb{P}\Bigg[\frac{{\Gamma}_{\text{SD}}^{(n)}}{{\Gamma}_{\text{SE}}^{(n)}}\le x\Bigg]
= \int_{0}^{\infty} F_{{\Gamma}^{(n)}_{\text{SD}}}(xy) f_{{\Gamma}_{\text{SE}}^{(n)}}(y) dy.
\end{align}
The asymptotic ESR for a selection scheme is then evaluated following (\ref{ergodic_secrecy_rate}) with the help of corresponding $F_{\Gamma_{\text{R}}^{(n^*)}}(x)$ derived from (\ref{gamma_r_high_snr}) in the BKU and BKA cases. In the OTS-BKA scheme, $F_{\Gamma_{\text{R}}^{(n^*)}}(x)$ in the high-SNR regime is derived following Section \ref{sec_distribution_snr_ratio} under the OTS-BKU case  using (\ref{cdf_SD}) and (\ref{PDF_SE}) in  (\ref{gamma_r_high_snr}) as
\begin{align}\label{gamma_ots_high_snr}
&F_{\Gamma_{\text{R}}^{(n^*)}}(x)=(1-s)+s\times\mathbb{P}\Bigg[\max_{n\in \{1,2\ldots, N\}} \Bigg\{\frac{{\Gamma}_{\text{SD}}^{(n)}}{{\Gamma}_{\text{SE}}^{(n)}} \Bigg\}\le x\Bigg] \nn \\
% &=(1-s)+s\Big[\int_{0}^{\infty}F_{{\Gamma}^{(n)}_{\text{SD}}}(yx)f_{{\Gamma}_{\text{SE}}^{(n)}}(y)dy\Big]^N\nn\\
% &=(1-s)+s\Big[\int_{0}^{\infty}\left(1-e^{-\lambda_{\text{D}} x y}\right)\left(\sum_{k=1}^{K}\frac{\prod_{m=1}^{K}\lambda_{\text{E}}^{(m)}}{\prod\limits_{\substack{j=1\\j\ne k}}^K(\lambda_{\text{E}}^{(j)}-\lambda_{\text{E}}^{(k)})} e^{-\lambda_{\text{E}}^{(k)} z}\right)dy\Big]^N \nn \\
% &=(1-s)+s\Big[1-\sum_{k=1}^{K}\frac{\prod_{m=1}^{K}\lambda_{\text{E}}^{(m)}}{(\lambda_{\text{D}} z+\lambda_{\text{E}}^{(k)}) \prod\limits_{\substack{j=1\\j\ne k}}^K(\lambda_{\text{E}}^{(j)}-\lambda_{\text{E}}^{(k)})}\Big]^N \nn \\
% &=1-\sum_{n=1}^{N}\binom{N}{n}(-1)^{n+1}s\Big(\sum_{k=1}^{K}\frac{\prod_{m=1}^{K}\lambda_{\text{E}}^{(m)} }{\lambda_{\text{D}}(z+a_k) \prod\limits_{\substack{j=1\\j\ne k}}^K(\lambda_{\text{E}}^{(j)}-\lambda_{\text{E}}^{(k)})}\Big)^n\nn\\
% &=1-\sum_{n=1}^{N}\binom{N}{n}(-1)^{n+1}s\sum_{i_1+\ldots+i_K=n}\binom{n}{i_1,\ldots,i_K}\prod_{k=1}^{K}\Bigg(\frac{\prod_{m=1}^{K}\lambda_{\text{E}}^{(m)}}{\lambda_{\text{D}}(x+a_t) \prod\limits_{\substack{j=1\\j\ne k}}^K(\lambda_{\text{E}}^{(j)}-\lambda_{\text{E}}^{(k)})}\Bigg)^{i_k}\nn\\
&=1-\sum_{n=1}^{N}\sum_{\mathbf{i}\in\mathcal{M}^{(n)}}\binom{N}{n}(-1)^{n+1}\binom{n}{i_1,\ldots,i_K}\nn\\
&\times\Bigg(\prod_{k=1}^{K}\Bigg(\frac{\prod_{m=1}^{K}\lambda_{\text{E}}^{(m)}}{\lambda_{\text{D}}\prod\limits_{\substack{j=1\\j\ne k}}^K(\lambda_{\text{E}}^{(j)}-\lambda_{\text{E}}^{(k)})}\Bigg)^{i_k}\Bigg)\frac{s}{\prod\limits_{t=1}^{K}\Big(x+\frac{{\lambda}_{E}^{(k)}}{\lambda_{\text{D}}}\Big)^{i_k}}.
\end{align}
The asymptotic ESR is then evaluated by following the same procedure as adopted for the OTS-BKU case in  (\ref{ERG_OTS_without}) from Section \ref{section_ergodic_secrecy_rate}. Therefore, using $F_{\Gamma_{\text{R}}^{(n^*)}}(x)$ from (\ref{gamma_ots_high_snr}) in (\ref{ergodic_secrecy_rate}), the asymptotic ESR is evaluated as 
\begin{align}\label{ERG_OTS_without_high}
&\mathcal{C}_{\text{erg}}^{\infty}
%& =\frac{1}{\ln(2)}\int_{1}^{\infty}\frac{1}{x}\sum_{n=1}^{N}\binom{N}{n}(-1)^{n+1}s\sum_{i_1+\ldots+i_K=n}\binom{n}{i_1,\ldots,i_K}\Bigg(\prod_{k=1}^{K}\Bigg(\frac{\prod_{m=1}^{K}\lambda_{\text{E}}^{(m)}}{\lambda_{\text{D}}\prod\limits_{\substack{j=1\\j\ne k}}^K(\lambda_{\text{E}}^{(j)}-\lambda_{\text{E}}^{(k)})}\Bigg)^{i_k}\Bigg)\nn\\
% &\times\prod_{k=1}^{K}\Big(\frac{1}{(x+a_t)^{i_k}}\Big)dx\nn\\
% &=\frac{1}{\ln(2)}\sum_{n=1}^{N}\binom{N}{n}(-1)^{n+1}s\sum_{i_1+\ldots+i_K=n}\binom{n}{i_1,\ldots,i_K}\Bigg(\prod_{k=1}^{K}\Bigg(\frac{\prod_{m=1}^{K}\lambda_{\text{E}}^{(m)}}{\lambda_{\text{D}}\prod\limits_{\substack{j=1\\j\ne k}}^K(\lambda_{\text{E}}^{(j)}-\lambda_{\text{E}}^{(k)})}\Bigg)^{i_k}\Bigg)\int_{1}^{\infty}\frac{dx}{x\prod_{k=1}^{K}(x+a_t)^{i_k}}\nn\\
% &=\frac{1}{\ln(2)}\sum_{n=1}^{N}\binom{N}{n}(-1)^{n+1}s\sum_{i_1+\ldots+i_K=n}\binom{n}{i_1,\ldots,i_K}\Bigg(\prod_{k=1}^{K}\Bigg(\frac{\prod_{m=1}^{K}\lambda_{\text{E}}^{(m)}}{\lambda_{\text{D}}\prod\limits_{\substack{j=1\\j\ne k}}^K(\lambda_{\text{E}}^{(j)}-\lambda_{\text{E}}^{(k)})}\Bigg)^{i_k}\Bigg)\Big(\int_{1}^{\infty}\frac{1}{x\prod_{k=1}^{K}a_t^{i_k}}dx\Big.\nn\\
% &\Big.+\int_{1}^{\infty}\sum_{k=1}^{K}\sum_{l_k=1}^{i_k}\frac{A_k^{(i_k)}}{(x+a_t)^{i_k-l_t+1}}dx\Big)\nn\\
=\frac{1}{\ln(2)}\sum_{n=1}^{N}\sum_{\mathbf{i}\in\mathcal{M}^{(n)}}\binom{N}{n}\binom{n}{i_1,\ldots,i_K}(-1)^{n+1}s\nn\\
&\times\Bigg(\prod_{k=1}^{K}\Bigg(\frac{\prod_{m=1}^{K}\lambda_{\text{E}}^{(m)}}{\lambda_{\text{D}}\prod\limits_{\substack{j=1\\j\ne k}}^K(\lambda_{\text{E}}^{(j)}-\lambda_{\text{E}}^{(k)})}\Bigg)^{i_k}\Bigg)\nn\\
&\times\sum_{k=1}^{K}\Bigg(\frac{\ln{\Big(\frac{{\lambda}_{E}^{(k)}}{\lambda_{\text{D}}}\Big)}}{\Big(\frac{{\lambda}_{E}^{(k)}}{\lambda_{\text{D}}}\Big)^{i_k}\prod\limits_{\substack{j=1\\j\ne k}}^{K}\Big(\frac{{\lambda}_{E}^{(k)}}{\lambda_{\text{D}}}-\frac{{\lambda}_{E}^{(j)}}{\lambda_{\text{D}}}\Big)^{i_ji_k}}+\sum_{l_k=1}^{i_k-1}A_k^{(i_k)}\nn\\
&\times\frac{\Big(1+\frac{{\lambda}_{E}^{(k)}}{\lambda_{\text{D}}}\Big)^{l_k-i_k} }{(i_k-l_k)}\Bigg).
\end{align}
% \color{black}
In the OTS-BKA scheme, $F_{\Gamma_{\text{R}}^{(n^*)}}(x)$ is  evaluated with the help of (\ref{cdf_SD}) and (\ref{PDF_SE}) in (\ref{gamma_r_high_snr}) in the high-SNR regime following the same procedure to derive (\ref{ERG_OTS_with}) in Section \ref{section_ergodic_secrecy_rate} for the OTS-BKU scheme. Using $F_{\Gamma_{\text{R}}^{(n^*)}}(x)$ in (\ref{ergodic_secrecy_rate}), the asymptotic ESR is written as 
\begin{align}\label{erg_asym_ots_with}
&\mathcal{C}_{\text{erg}}^{\infty}=\frac{1}{\ln(2)}\sum_{n=1}^{N}\sum_{\mathbf{i}\in\mathcal{M}^{(n)}}\sum_{k=1}^{K}\binom{N}{n}\binom{n}{i_1,\ldots,i_K}(-1)^{n+1}s^n\nn\\
&\times\sum_{k=1}^{K}\Bigg(\frac{\ln{\Big(\frac{{\lambda}_{E}^{(k)}}{\lambda_{\text{D}}}\Big)}}{\Big(\frac{{\lambda}_{E}^{(k)}}{\lambda_{\text{D}}}\Big)^{i_k}\prod\limits_{\substack{j=1\\j\ne k}}^{K}\Big(\frac{{\lambda}_{E}^{(k)}}{\lambda_{\text{D}}}-\frac{{\lambda}_{E}^{(j)}}{\lambda_{\text{D}}}\Big)^{i_ji_k}}+\sum_{l_k=1}^{i_k-1}A_k^{(i_k)}\nn\\
&\times\frac{\Big(1+\frac{{\lambda}_{E}^{(k)}}{\lambda_{\text{D}}}\Big)^{l_k-i_k} }{(i_k-l_k)}\Bigg).   
\end{align}
We notice from (\ref{ERG_OTS_without_high}) and (\ref{erg_asym_ots_with}) that the asymptotic ESR in both the BKU and BKA cases depends on the ratio of destination link SNR to eavesdropping link SNR. Thus, as this ratio improves, the asymptotic ESR also improves.

It is difficult to represent (\ref{erg_asym_ots_with}) in the form of the asymptotic straight line as shown in (\ref{eq_straight_line_form_ESR_MINESBKU}). Therefore we show it for $K=1$. In this case, the asymptotic ESR of the OTS-BKU scheme simplifies to
\begin{align}\label{erg_asym_ots_k1}
\mathcal{C}_{\text{erg}}^{\infty}
% =\frac{1}{\ln(2)}\sum_{n=1}^{N}\binom{N}{n}(-1)^{n+1}s\lb(\log\lb(\frac{1}{\lambda_{\text{D}}}\rb)+\log\lb({\lambda_{E}^{(1)}}\rb)+\sum_{j=2}^{n}\frac{1}{j-1}\rb)\nn\\
&=\frac{1}{\ln(2)}\sum_{n=1}^{N}\binom{N}{n}(-1)^{n+1}s\lb(\log\lb(\frac{\lambda_{E}^{(1)}}{\lambda_{\text{D}}}\rb)-H_{n-1}\rb),  
\end{align}
where $H_{n-1}$ denotes the $(n-1)$-th harmonic number. Thus, we find $S^{\infty}=\frac{s}{\ln(2)}$
% \begin{align}\label{s_inf_ots_without}
% S^{\infty}&=\frac{s}{\ln(2)},
% \end{align}
and
\begin{align}\label{l_inf_ots_without}
L^{\infty}
% &=\log\Big(\frac{1}{\lambda_{E}^{(1)}}\Big)-\sum_{n=1}^{N}\binom{N}{n}\frac{(-1)^{n+1}}{1-j}\nn\\
&=\log\Big(\frac{1}{\lambda_{E}^{(1)}}\Big)+\sum_{n=1}^{N}\binom{N}{n}{(-1)^{n+1}}{{H}_{n-1}}.
\end{align}
Similarly, for the OTS-BKA case when $K=1$, we can show $S^{\infty}=\frac{1-(1-s)^N}{\ln(2)}$
% \begin{align}\label{s_inf_ots_with}
% S^{\infty}&=\frac{1-(1-s)^N}{\ln(2)}
% \end{align}
and
\begin{align}\label{l_inf_ots_with}
L^{\infty}
% &=\log\Big(\frac{1}{\lambda_{E}^{(1)}}\Big)-\frac{1}{1-(1-s)^N}\sum_{n=1}^{N}\sum_{j=2}^{n}\binom{N}{n}\frac{(-1)^{n+1}s^n}{1-j}\nn\\
&=\log\Big(\frac{1}{\lambda_{E}^{(1)}}\Big)+\sum_{n=1}^{N}\binom{N}{n}\frac{{(-1)^{n+1}s^n}}{1-(1-s)^N}{H_{n-1}}.
\end{align}
We note here that the slope and the offset parameter
% $L^{\infty}$ in (\ref{l_inf_ots_with}) 
for the BKA case for $K=1$ matches with the result shown in \cite{chinmoy_letter2021}.
We observe that the slope of the OTS-BKU for $K=1$
is the same as the slope of MIN-ES-BKU and TTS-BKU cases. The slope of the OTS-BKA case for $K=1$ is also the same as the slope of MIN-ES-BKA and TTS-BKA cases. Thus, irrespective of the selection schemes, the slope is the same for the BKU case, and it is also the same for the BKA case. 
% It is also seen that irrespective of the selection schemes, the slope is the same for the BKU case, and it also is the same for the BKA case. 
However, $L^\infty$ is different for each selection scheme and in the BKU and BKA cases. Similar to the observation in the MIN-ES and TTS  schemes, $L^\infty$ in the OTS scheme does not depend on $s$ and $1/\lambda_{\text{D}}$ for the BKU cases and  depends on $s$, $N$, $K$, and $1/\lambda_{\text{E}}^{(k)}$ for the BKA cases. 
{For $K=1$, $L^\infty$ in the OTS scheme in (\ref{l_inf_ots_without}) and (\ref{l_inf_ots_with}) is smaller than $L^\infty$ of the MIN-ES and TTS scheme in (\ref{l_infinity_min_without}), (\ref{l_infinity_min_with}), (\ref{l_inf_tts_without}), and (\ref{l_inf_tts_with}). This is why OTS scheme performs the best.}

The asymptotic ESR expression with the help of the slope and offset clearly shows how it depends on  $s$, $1/\lambda_{\text{D}}$, $1/\lambda_{\text{E}}^{(k)}$ for all $k\in\{1,\ldots,K\}$, $K$, and $N$
which was difficult to understand from the exact ESR expressions.
% in (\ref{erg_min_es_without_final}), (\ref{erg_min_es_with}), (\ref{ERG_TTS_without}), and (\ref{ERG_TTS_with}). 
Further, it is to be noted that the asymptotic ESR depends on parameters $s$, $1/\lambda_{\text{D}}$, $1/\lambda_{\text{E}}^{(k)}$ for all $k\in\{1,\ldots,K\}$, $K$, and $N$, in contrast, the asymptote of the SOP and NZSR depends only on $s$ for the BKU case and only on $s$ and $N$ for the BKA case.

\section{Numerical Results} \label{section_numerical_results}
This section describes the numerical along with simulated results.  Throughout this section, we assume that the threshold secrecy rate is set as $R_{th}=1$ bits per channel use (bpcu), {$K=3$, and the eavesdroppers' SNRs are set at $1/\lambda_{\text{E}}^{(1)}=6$ dB,  $1/\lambda_{\text{E}}^{(2)}=9$ dB, $1/\lambda_{\text{E}}^{(3)}=13$ dB  for each $n\in\{1,\ldots, N\}$}, 
% (\textcolor{red}{this is not what you want to say according to the other similar notations in set}) 
unless stated otherwise.
Results are plotted for the SOP  
and ESR for each selection scheme in BKU (designated by black color) and BKA (designated by red color) cases. The NZSR is not shown as it can be easily obtained from the SOP and the observations are similar to that of the SOP.
The simulation is denoted by `$\times$', and the horizontal solid line denotes the asymptote from Fig \ref{fig_sop_s} to Fig \ref{fig_erg_nk}, and the slanted straight lines denote the asymptote in Fig \ref{fig_erg_asym}. 
% Each graph is plotted for both cases i) when backhaul link activity knowledge is available, designated by red color, and ii) when the knowledge of backhaul link activity is unavailable, designated by black color.
The validity of our analysis is evident from the figures, as the numerical and simulation results match perfectly.
We notice from all the figures that the OTS scheme outperforms all other schemes in both the BKU and BKA cases. Further, it is noticed that the BKA cases outperform the BKU cases corresponding to their associated selection schemes, as the backhaul link activity knowledge improves the secrecy performance. 

\subsection{Secrecy Outage Probability}
\begin{figure}
\centering 
\includegraphics[width=3.6in]{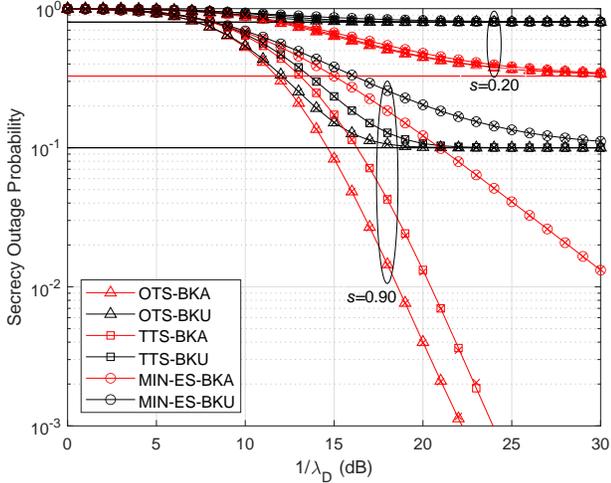} 
\caption{Variation of SOP with ${1}/{\lambda_{\text{D}}}$  for different values of $s$.}
% \vspace{-0.5cm}
\label{fig_sop_s}
\end{figure}

In Fig. \ref{fig_sop_s}, the SOP has been plotted versus destination channel SNR 1/$\lambda_{\text{D}}$ for different values of $s = \{0.20, 0.90\}$ with $N=5$ and $K=3$. We observe that the secrecy performance improves with the increase in SNR $1/\lambda_{\text{D}}$ till it reaches its asymptotic value. We also notice that the performance improves with the increase in $s$. Further at low SNRs, the MIN-ES scheme performs better than the TTS scheme. After a certain value of 1/$\lambda_{\text{D}}$, the TTS scheme overcomes the MIN-ES scheme. This is because selecting the worst eavesdropping link is advantageous over selecting the best destination link among already degraded destination links.
In contrast, selection among the destination links is advantageous at high SNR when the links are usually good. At high SNRs, the performance difference between the TTS scheme and the MIN-ES scheme is also more significant than the difference between the OTS scheme and the TTS scheme. 
It is also observed that the SOP saturates to a constant value at  high SNR, which is represented by the asymptotic straight line. As $s$ increases, saturation levels decrease. Further, the BKA case has a lower saturation level than the BKU case. This can also be confirmed from the numerical analysis in Section \ref{sec_asymptotic_analysis_sop_unreliable}.

\begin{figure}
\centering 
\includegraphics[width=3.6in] {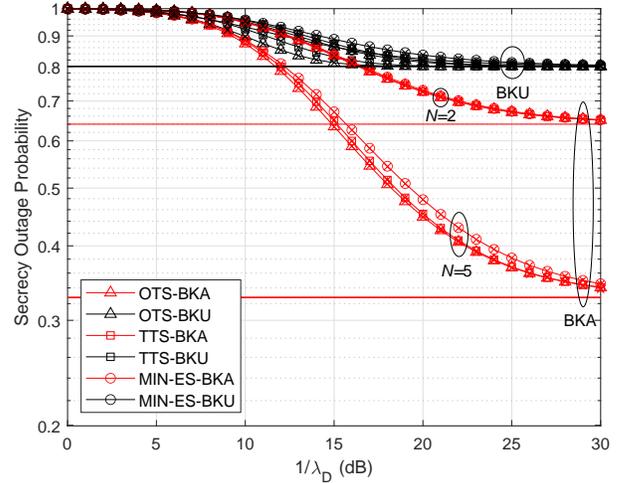}  
\caption{Variation of SOP with ${1}/{\lambda_{\text{D}}}$  for different values of $N$.}
% \vspace{-0.3cm}
\label{fig_sop_n}
\end{figure}
Fig. \ref{fig_sop_n} shows the SOP versus 1/$\lambda_{\text{D}}$ for different values of $N = \{2,  5\}$ when $s = 0.20$ and $K=3$. We observe that the SOP decreases for all considered schemes as the number of transmitters increases. More the number of transmitter choices, the better the secrecy performance of the system. However, the rate of performance improvement is far slower in the case of BKU. Thus, we notice that the secrecy performance of the system with fewer transmitters, i.e., $N=2$ for BKA cases, outperforms the BKU cases with a large number of transmitters, i.e., $N=5$ for any selection scheme. The SOP for the BKU case for all the selection schemes saturates to the same constant value irrespective of $N$ as it depends only on $s$, however,  for the BKA case, the SOP for all the selection schemes saturates to different levels as it depends both on $s$ and $N$.

\begin{figure}
\centering 
\includegraphics[width=3.6in]{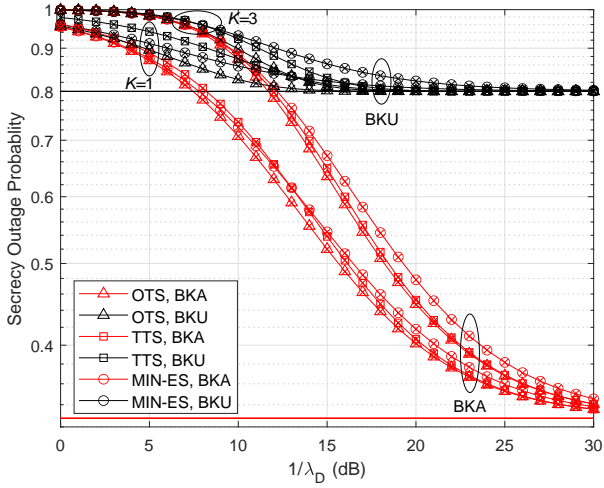} 
\caption{ Variation of SOP with ${1}/{\lambda_{\text{D}}}$ for different values of $K$.}
\label{fig_sop_k}
\end{figure}

Fig. \ref{fig_sop_k} plots the SOP versus 1/$\lambda_{\text{D}}$ for the proposed selection schemes for different number of eavesdroppers $K=\{1,3\}$ when $s=0.90$ and  $N=5$. The eavesdroppers' SNRs  are set at $1/\lambda_{\text{E}}^{(1)}=6$ dB,  $1/\lambda_{\text{E}}^{(2)}=9$ dB, and $1/\lambda_{\text{E}}^{(3)}=13$ dB when $K=3$ for each $n\in\{1,\ldots, N\}$, and $1/\lambda_{\text{E}}^{(1)}=13$ dB when $K=1$.  We note that the secrecy performance degrades with the increased number of eavesdroppers.  
% The performance of the selection schemes follows the usual trend where the performance from the best to the worst is OTS,TTS, and the MIN-ES scheme. 
The MIN-ES scheme performs better than the TTS at low SNRs, and the TTS scheme outperforms the MIN-ES scheme at high-SNR, which was also found in Fig. \ref{fig_sop_s}. We find that irrespective of the selection schemes, the asymptotes for both the values of $K$ coincide for a given backhaul link activity knowledge case (BKU and BKA). This shows that the asymptotes are independent of $K$ and the selection schemes and only depend on whether the backhaul link activity knowledge is available or not.   This confirms the analysis for asymptotic SOP in Section \ref{sec_asymptotic_analysis_sop_unreliable}. 

% \textcolor{red}{In all the above figures, it is noticed that the SOP saturates to a constant value at  high SNR, which is presented by the asymptotic straight line. For the BKU case, it depends only on $s$, and for the BKA case, it depends on $s$ and $N$. This can also be confirmed from the numerical analysis in Section \ref{sec_asymptotic_analysis_sop_unreliable}.}

\subsection{Ergodic Secrecy Rate}
\begin{figure}
\centering 
\includegraphics[width=3.6in]{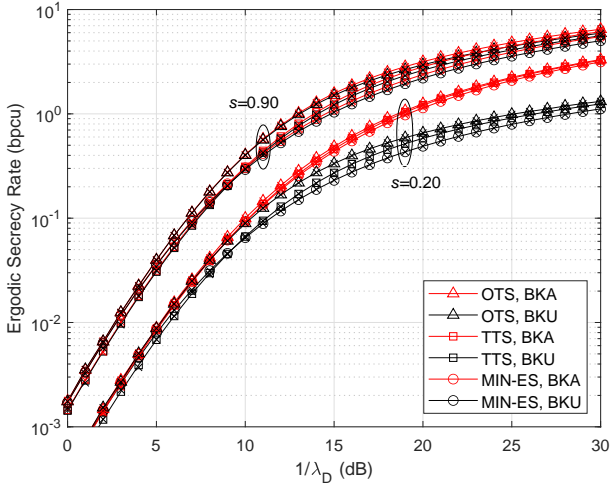} 
\caption{Variation of ESR with ${1}/{\lambda_{\text{D}}}$  for different values of $s$.}
\label{fig_erg_s}
\end{figure}

Fig. \ref{fig_erg_s} plots the ESR versus  $1/\lambda_{\text{D}}$  for two different values of the backhaul link reliability factor, $s=\{0.90, 0.20\}$ with $N=5$ and $K=3$.
% As 1/$\lambda_{\text{D}}$ increases, the ESR improves until it reaches its asymptotic value.
The ESR performance of the transmitter selection schemes follow a similar trend of the SOP, e.g., the MIN-ES scheme performs better than the TTS scheme at low SNRs, and the ESR improves as $s$ increases. However, it is noticed that the ESR improvement from the BKU to BKA case by utilizing the backhaul activity knowledge when $s$ is low is more significant as compared to the case when $s$   high. This implies that the knowledge of backhaul link activity can significantly improve the secrecy performance when the backhaul link reliability is relatively low.

\begin{figure}
\centering 
\includegraphics[width=3.6in]{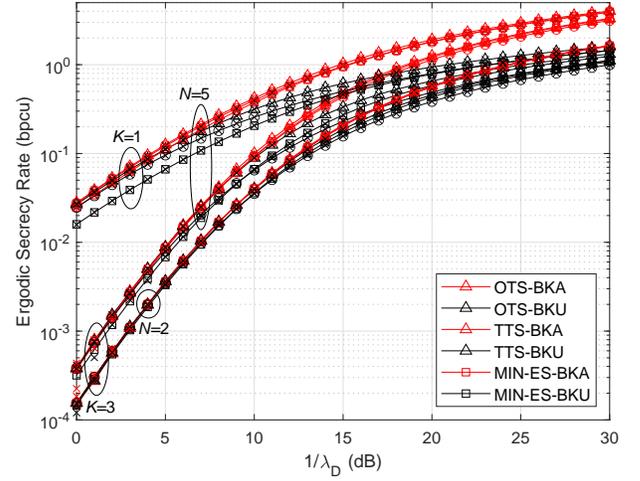} 
\caption{Variation of ESR with ${1}/{\lambda_{\text{D}}}$  for different values of $N$ and $K$.}
\label{fig_erg_nk}
\end{figure}

Fig. \ref{fig_erg_nk} shows the effect of a different number of eavesdroppers $K=\{1, 3\}$ and $N=\{2,5\}$ on the ESR  for all the selection schemes when $s=0.2$. The eavesdroppers' SNRs  are set at $1/\lambda_{\text{E}}^{(1)}=6$ dB,  $1/\lambda_{\text{E}}^{(2)}=9$ dB, and $1/\lambda_{\text{E}}^{(3)}=13$ dB when $K=3$ for each $n\in\{1,\ldots, N\}$, and $1/\lambda_{\text{E}}^{(1)}=13$ dB when $K=1$. We notice that the ESR performance improves when the number of transmitters increases.  However, the system's secrecy performance decreases as the number of eavesdroppers increases.

\begin{figure}
\centering 
\includegraphics[width=3.6in]{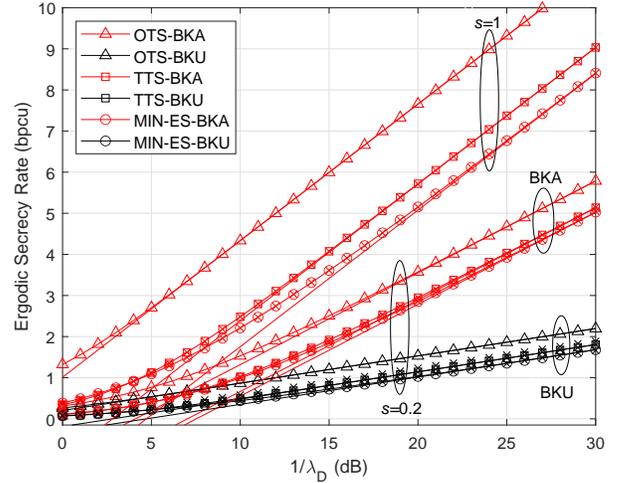} 
\caption{Variation of ESR with ${1}/{\lambda_{\text{D}}}$ for perfect and unreliable backhaul links.}
\label{fig_erg_asym}
\end{figure}

In Fig. \ref{fig_erg_asym}, the ESR and its asymptotic values are plotted versus $\frac{1}{\lambda_{\text{D}}}$ when $N=5$ and $K=1$ for the perfect and unreliable backhaul (BKU and BKA) cases.
% The asymptotic ESR is plotted for the selection schemes.
The ESR axis is shown in the linear scale to plot the asymptotes as straight lines following (\ref{eq_straight_line_form_ESR_MINESBKU}). The asymptotic straight lines match well with the exact ESRs at the high-SNR regime. We notice that the perfect backhaul case, i.e., $s=1$, outperforms the unreliable backhaul case, i.e., $s<1$, which is intuitive.   It also has the highest slope.  However, as $s$ decreases from $s=1$, the slope of the ESR in the unreliable case also decreases. Further, among unreliable backhaul cases, the BKA case has a better slope than the BKU case. This signifies that the rate of improvement of ESR in the BKA case with SNR is better than in the BKU case due to the available backhaul link activity knowledge.
It is also noticed that the slope for the BKU cases, irrespective of the selection schemes, is the same.  So does the slope of the BKA cases for all the selection schemes. This shows that the slope does not depend on the selection schemes only depends on the backhaul link activity knowledge and the reliability factor. This can also be confirmed from the derived slopes in the asymptotic analysis in Section \ref{section_asymptotic_ESR}.

\section{Conclusions}  \label{conclusions} 
We adopt a generalized methodology to include the backhaul reliability factor in the secrecy performance analysis of transmitter selection schemes against colluding eavesdroppers and provide a uniform approach to obtain exact closed-form NZSR, SOP, and ESR expressions from the distribution of the ratio of destination link and eavesdropping link SNR irrespective of selection schemes.  Simplified asymptotic expressions are provided for each transmitter selection scheme to understand the effect of system parameters and available backhaul activity knowledge. We observe that the asymptotic NZSR and SOP performance can not be improved by increasing the number of transmitters when the backhaul link activity knowledge is unavailable and saturates to a value depending on the backhaul reliability factor only. In contrast, asymptotic saturation can be improved by increasing the number of transmitters when the backhaul link activity knowledge is available. In both cases, the asymptotic saturation value does not depend on the number of eavesdroppers. 
The effect of unavailable backhaul activity knowledge degrades the rate of improvement of the ESR with SNR. Though the number of eavesdroppers has no effect on the rate of improvement of the ESR with SNR, the ESR degrades with the increase in the number of eavesdroppers. We notice that the knowledge of backhaul link activity can significantly improve the secrecy performance when the backhaul link reliability is relatively low. We also find that the MIN-ES scheme is better at low SNR than the TTS scheme, however, the TTS scheme outperforms it when SNR improves.

\bibliographystyle{IEEEtran}
\bibliography{IEEEabrv,BACKHAUL_BW}

\end{document}